\documentclass[journal]{IEEEtran}
\onecolumn
\usepackage{mathrsfs}
\usepackage{amsmath}
\usepackage{amsthm}
\usepackage{graphicx}
\usepackage{amssymb}
\usepackage{epstopdf}
\usepackage{enumerate}
\usepackage{longtable,tabularx,float,booktabs,bigstrut}
\usepackage{cite}
\usepackage{mathcomp}
\usepackage{supertabular}
\usepackage{stmaryrd}
\usepackage{color}
\usepackage{url}
\usepackage{makecell}
\usepackage[OT2,OT1]{fontenc}
\usepackage{bm}
\usepackage{caption}
\usepackage{multirow}
\usepackage{hyperref}
\newtheorem{corollary}{Corollary}[section]

\newtheorem{lemma}{Lemma}[section]
\newtheorem{theorem}{Theorem}[section]

\newtheorem{example}{Example}[section]
\newtheorem{construction}{Construction}[section]

\newtheorem{remark}{Remark}[section]

\def\ie{{\it i.e.\ }}

\newcommand{\B}{\mathcal {B}}
\newcommand{\C}{\mathcal {C}}
\newcommand{\G}{\mathcal {G}}

\renewcommand{\S}{\mathcal {S}}

\newcommand{\bbZ}{{\mathbb Z}}

\begin{document}

\title{Optimal codes with small constant weight in $\ell_1$-metric}

\author{Tingting~Chen,~Yiming~Ma,~and~Xiande~Zhang
        \thanks{T. Chen ({\tt ttchenxu@mail.ustc.edu.cn}) and Y. Ma ({\tt mym024@mail.ustc.edu.cn}) are with
        School of Cyber Security, University of Science and Technology of China, Hefei, 230026, Anhui, China.}

\thanks{X. Zhang ({\tt drzhangx@ustc.edu.cn}) is with School of Mathematical Sciences,
University of Science and Technology of China, Hefei, 230026, Anhui, China.  The research of X. Zhang is supported by NSFC under grant 11771419, and by ``the  Fundamental
Research Funds for the Central Universities''.}
}
\maketitle

\begin{abstract} Motivated by the duplication-correcting problem for data storage in live DNA, we study the construction of constant-weight codes in $\ell_1$-metric. By using packings and group divisible designs in combinatorial design theory, we give constructions of optimal codes over non-negative integers and optimal ternary codes with $\ell_1$-weight $w\leq 4$ for all possible distances. In general, we derive the size of the largest ternary code with constant weight $w$ and distance $2w-2$ for sufficiently large length $n$ satisfying $n\equiv 1,w,-w+2,-2w+3\pmod{w(w-1)}$.
\end{abstract}
\begin{IEEEkeywords}
\boldmath DNA storage, tandem duplication error, constant-weight code, $\ell_1$-metric,  packing.
\end{IEEEkeywords}

\section{Introduction}
\IEEEPARstart{C}{onstant-weight} codes (CWCs) with Hamming distance constraint have attracted a lot
attention in recent years due to their vast
applications, such as in coding for bandwidth-efficient channels \cite{costello2007channel} and the design of oligonucleotide
sequences for DNA computing \cite{king2003bounds,milenkovic2005design}. One of the central problems in their study is to
 determine the maximum size of  CWCs
due to their close relations to combinatorial design theory, see for example \cite{graham1980lower,brouwer1990new,agrell2000upper,zhang2010,zhang2012optimal,chee2014complexity,chee2015hanani,
chee2017linear,chee2017constructions,chee2019decompositions}. Although there are several different metrics which have been considered in coding theory, to the best of our knowledge, there is little known  for CWCs in the literature besides Hamming distance.

In this paper, we initiate the study of  CWCs with $\ell_1$-metric, which is motivated from the error correcting problem of  data storage in live DNA \cite{jain2017duplication}. To prove the reliability of information stored in live DNA, codes which can correct errors such as tandem duplication, point mutations, insertions, and deletions arising from various mutations, must  be considered. Among these, duplication-correcting codes have been studied by a number of recent works, see \cite{jain2017duplication,kovavcevic2018asymptotically,lenz2018bounds,tang2019single,yehezkeally2019reconstruction}. In \cite{jain2017duplication},  the authors studied tandem duplication, which is a process of inserting a copy of a segment of the DNA adjacent to its original position. For example, for a sequence $AGCTCT$, $CTCT$ is a tandem duplication error of length two on $CT$.
Tandem duplications  constitute
about 3\% of the human genome \cite{lander2001initial} and may cause important
phenomena such as chromosome fragility, expansion diseases,
silencing genes \cite{usdin2008biological}, and rapid morphological variation \cite{fondon2004molecular}.
Jain et al. \cite{jain2017duplication} proposed a coding scheme to combat tandem duplications, which is based on CWCs in $\ell_1$-metric over integers. More specifically \cite[Construction B]{jain2017duplication}, given a set of $\ell_1$-metric codes, which are of length $m$, constant $\ell_1$-weight $w$ and minimum $\ell_1$-distance $2(t+1)$, with $1\leq m\leq n-k+1$ and $0\leq w\leq \lfloor\frac{n-k}{k}\rfloor$,  one can construct a code of length $n$ capable of correcting $t$ errors of $k$-tandem-duplication. Due to this, a choice of optimal $\ell_1$-metric codes with certain weights and lengths will result in an optimal code correcting tandem duplications.

Codes in $\ell_1$-metric distance have been widely studied because of their applications in rank-modulation scheme for flash memory \cite{barg2010codes,jiang2009error,tallini20111,zhou2014systematic,farnoud2013error,kabatianskycodes}. However,  most works focus on permutation codes or multi-permutation codes. Kova\v{c}evi\'{c} and Tan \cite{vincent2018code} gave some properties and constructions of \emph{multiset codes}, based on Sidon sets and lattices, and derived bounds on the size of such codes. The multiset codes they studied are indeed the constant-weight codes in  $\ell_1$-metric over non-negative integers, in the context that the weight grows, the length is either constant or grows proportionally, while the distance remains fixed. Their construction is asymptotically optimal for small lengths, or for small distances.
Jinushi and Sakaniwa \cite{jinushi1990construction} proposed a construction for error-correcting codes in $\ell_1$-metric which relies upon the properties of generalized Hadamard matrices \cite{butson1962generalized}. They used the term \emph{absolute
summation distance} that we believe is \emph{constant weight in  $\ell_1$-metric} (Unfortunately, we could not find a copy of their paper to confirm it).

In \cite{zhang2010,zhang2012optimal}, the authors investigated the construction of optimal ternary CWCs with Hamming distance constraint, by using the tool of group divisible codes, which is an analog of group divisible designs (GDDs) in combinatorial design theory \cite{ge2007group}. Motivated by their methods, in this work, we construct CWCs  in  $\ell_1$-metric by using packings and GDDs for fixed weight $w$ and distance $d$, and determine the size of the largest codes of weight $w\leq 4$ for all possible distance $d$ and length $n$. For $w=1$ or $2$, codes are trivial. Our main contributions are listed below.

\begin{enumerate}
  \item For CWCs of weight $w\in \{3,4\}$ over non-negative integers, we derive the size $A(n,d,w)$ of the largest code for all distance $d$ and length $n$ completely. Some codes are constructed by optimal packings of triples by pairs \cite{stinson2006packings}, or by optimal packings of quadruples by triples \cite{bao2015completion}. In particular, we show that $A(n,4,3)\sim\frac{n^2}{6}, A(n,4,4)\sim\frac{n^3}{24}$ and $A(n,6,4) \sim \frac{n^2}{12} $, which improve the bounds given in \cite{vincent2018code} for $\frac{n^2}{6}\lesssim A(n,4,3)\lesssim \frac{n^3}{6}$, $\frac{n^3}{24}\lesssim A(n,4,4)\lesssim \frac{n^4}{24}$, and $A(n,6,4)\lesssim \frac{n^4}{24}$\footnote{$f(n)\sim g(n)$ means $\lim_{n\rightarrow \infty}f(n)/g(n)=1$, and $f(n)\lesssim g(n)$ means $\liminf_{n\rightarrow \infty}g(n)/f(n)\geq 1$.}.
  \item  For CWCs  over $\{0,1,2\}$, that is, ternary CWCs,   we construct non-trivial codes with maximum size by using Steiner triple systems \cite{Colbourn2006} and  packings with special leave graphs \cite{colbourn1986quadratic,colbourn1998class}, and solve the case completely for $w=3$. For $w=4$, the maximum sizes of codes are determined based on GDDs, but with very few   cases unsolved, for which we provide upper and lower bounds.
  \item For ternary CWCs with $d=2w-2$, we give a general construction using a result on  graph packings  of Alon et al. \cite{alon1998packing}, and determine the size of the largest code for sufficiently large length $n$ satisfying $n\equiv 1,w,-w+2,-2w+3\pmod{w(w-1)}$.
\end{enumerate}



Our paper is organized as follows. In Section~\ref{pre}, we give  necessary definitions, notations, results in combinatorial design theory, and the connections between $\ell_1$-metric codes and packings. In Section~\ref{cwcz}, we construct optimal CWCs over non-negative integers for weight three and four. In Sections~\ref{cwc3} and \ref{conca}, we consider ternary CWCs, and give combinatorial constructions for optimal codes for weight  three and four, respectively. In Section~\ref{dis2w}, we deal with  ternary codes for general weight $w$ and distance $2w-2$ by graph packings. Finally, we conclude our results in Section~\ref{s:con}.

\section{Preliminaries}\label{pre}

Let $\bbZ_{\geq 0}$ denote the set  of non-negative integers and $\bbZ_q$ denote the ring of integers modulo $q$, for any integer $q\geq 2$.  For any integers $a<b$, let $[a,b]$ denote the set of integers $\{a,a+1,\ldots,b\}$. We further abbreviate $[1,b]$ as $[b]$. For two sets $A$ and $B$, the {\it symmetric difference} of $A$ and $B$ is the set of elements which are in either of the $A$ or $B$ and not in their intersection, denoted by $A\vartriangle B$.


\subsection{CWCs with $\ell_1$-metric}

A $q$-ary code $\C$ of length $n$ is a set of vectors in $I_q^n$, where $I_q:=\{0,1,\ldots,q-1\}\subset \bbZ_{\geq 0}$. The elements of $\C$ are called {\it codewords}. For any two codewords $\textbf{u}=(\textbf{u}_1,\ldots,\textbf{u}_n),\textbf{v}=(\textbf{v}_1,\ldots,\textbf{v}_n)\in\C$, the {\it support} of $\textbf{u}$ is defined as $supp(\textbf{u})=\{x\in [n]\mid \textbf{u}_x\neq 0\}$, and the {\it $\ell_1$-distance} between $\textbf{u}$ and $\textbf{v}$ is defined as $d_{l_1}(\textbf{u},\textbf{v})=\sum_{x\in [n]}|\textbf{u}_x-\textbf{v}_x|$ (computations are over the ring of integers). The {\it $\ell_1$-weight} of $\textbf{u}$ is defined as the $\ell_1$-distance of $\textbf{u}$ and the zero vector, \ie\,wt$_{l_1}(\textbf{u})=\sum_{x\in [n]}|\textbf{u}_x|$.
A code $\C$ is said to be of {\it constant-weight $w$} if wt$_{l_1}(\textbf{u})=w$ for every codeword $\textbf{u}\in\C$, and of {\it minimum $\ell_1$-distance $d$} if $d_{l_1}(\textbf{u},\textbf{v})\geq d$ for any two distinct codewords $\textbf{u},\textbf{v}\in \C$. If both properties are satisfied, then a code is called a \emph{constant-weight code} in $\ell_1$-metric and denoted by an $(n,d,w)_q$ code if it is $q$-ary, and an $(n,d,w)$ code if it is over $\bbZ_{\geq 0}$. For convenience, we call a codeword of \emph{type $1^{e_1}2^{e_2}\cdots (q-1)^{e_{q-1}}$} if it has exactly $e_i$ positions with entry $i$ for $i\in[q-1]$, and all the rest positions are $0$.

In \cite{jain2017duplication}, Jain et al. established a connection between codes capable of  correcting tandem duplications and constant-weight codes with $\ell_1$-metric. They showed that a code is able to correct $t$ tandem duplications of length $k$ if and only if the zero signatures of the $z$-part of all $k$-congruent codewords form a constant-weight code in $\ell_1$-metric over integers with distance at least $2(t+1)$ ( see \cite[Theorem 20]{jain2017duplication} for details). More importantly, a choice of optimal $\ell_1$-metric constant-weight codes for certain wights and lengths will result in optimal tandem  duplication correcting codes \cite[Construction B]{jain2017duplication}. Further if we assume that each segment of length $k$ is duplicated at most $q-1$ times,  then we need  $q$-ary constant-weight codes in $\ell_1$-metric.


Motivated by this connection, we consider  constant-weight codes in $\ell_1$-metric with maximum possible size. Since we only consider $\ell_1$-metric in this paper, we omit the subscript $\ell_1$ or the term $\ell_1$-metric unless otherwise specified.
 The maximum number of codewords  in an $(n,d,w)_q$ code is denoted by $A_q(n,d,w)$, and the $(n,d,w)_q$ code is called {\it optimal} if it has $A_q(n,d,w)$ codewords. Similarly for codes over $\bbZ_{\geq 0}$, we use $A(n,d,w)$ to denote the largest possible number of codewords.

In the remaining of this paper, we  focus on determining the values of $A_q(n,d,w)$ and $A(n,d,w)$ by constructing optimal CWCs. The followings are some trivial cases.

\begin{theorem}\label{base}
  \begin{itemize}
     \item[(a)] $A_q(n,2\delta-1,w)=A_q(n,2\delta,w)$; $A(n,2\delta-1,w)=A(n,2\delta,w)$.
     \item[(b)] $A_q(n,2\delta,w)=A(n,2\delta,w)=1$ if $w<\delta.$
     \item[(c)] $A_q(n,2w,w)=\left\lfloor \frac{n}{\lceil\frac{w}{q-1}\rceil}\right\rfloor$; $A(n,2w,w)=n.$
     \item[(d)] $A_q(n,2,w)=A(n,2,w)=\binom{n+w-1}{w}$ if $w\leq q-1$; $A_q(n,2,w)=\sum_{j=0}^t(-1)^j\binom{n}{j}\binom{n-1+w-jq}{w-jq}$ if $w>q-1$, where $t=\lfloor \frac{w}{q}\rfloor$.
   \end{itemize}
\end{theorem}

\begin{proof}
  The equalities in (a) follow because the $\ell_1$-distance between any two codewords of constant weight is even. The results in (b) are obvious, and results in  (c) follow because the codewords must have disjoint supports and the sizes of supports of codewords over $I_q$ and $\bbZ_{\geq 0}$ are at least $\lceil\frac{w}{q-1}\rceil$ and one, respectively.

  For (d), the value of $A(n,2,w)$ equals the number of distinct vectors of length $n$ and weight $w$, that is the number of non-negative integer solutions to the equation $x_1+\cdots+x_n=w$, which is $\binom{n+w-1}{w}$. For $q$-ary codes, we split them in two cases. If $w\leq q-1$, any entry of the codeword is at most $w$, which is the same as codes over non-negative integers, then $A_q(n,2,w)=A(n,2,w)=\binom{n+w-1}{w}$. If $w>q-1$, any entry of the codeword is at most $q-1$. Let $a_j$ be the  number of distinct $q$-ary vectors of length $n$ and  weight $j$, that is, $a_j$ equals the number of integer solutions to the equation $x_1+\cdots+x_n=w$, where $0\leq x_i\leq q-1$. Then the generating function of the sequence $a_j$ is
  \begin{equation*}
    (1+x+\cdots+x^{q-1})^n=\sum_{j=0}^{kn}a_jx^j.
  \end{equation*}
  Expand the above equation, we get $A_q(n,2,w)=a_w=\sum_{j=0}^t(-1)^j\binom{n}{j}\binom{n-1+w-jq}{w-jq}$, where $t=\lfloor \frac{w}{q}\rfloor$.
\end{proof}

\vspace{0.5cm}
By Theorem~\ref{base},  we only need to consider the even distances between $4$ and $2w-2$ for any code of constant weight $w$.


\subsection{Designs}\label{sec:iib}

A {\it set system} is a pair $\S=(X,\B)$, where $X$ is a finite set of {\it points} and $\B$ is a set of subsets of $X$, called {\it blocks}. The {\it order} of $\S$ is the number of points $|X|$, and the {\it size} of $\S$ is the number of blocks $|\B|$.

A {\it graph} $G$ is a set system $(V,E)$ with all blocks in $E$ being $2$-subsets of $V$, in which a point of $V$ is called a {\it vertex} and a block of $E$ is called an {\it edge}. The {\it degree} of a vertex $v$ is the number of edges containing $v$. Let $ \delta(G)$ denote the minimum vertex degree of $G$. A graph is called  {\it complete} if each pair of vertices is connected by an edge, and denoted by $K_n$ if $|V|=n$, such a graph is also called {\it clique}. We call a sequence $(v_1,v_2,\ldots,v_m)$ of distinct vertices a {\it cycle} of length $m$ if $\{v_i,v_{i+1}\}\in E$ for all $i\in[m-1]$ and $\{v_m,v_1\}\in E$.

Let $K$ be a set of positive integers. A $t$-$(n,K,1)$ packing is a set system $(X,\B)$ of order $n$, such that $|B|\in K$ for each $B\in \B$, and every $t$-subset of points occurs in at most one block of $\B$. When $K=\{k\}$, we just write $k$ instead of $\{k\}$. The {\it packing number} $D(n,k,t)$ is the maximum number of blocks in any $t$-$(n,k,1)$ packing. A $t$-$(n,k,1)$ packing $(X,\B)$ is {\it optimal} if $|\B|=D(n,k,t)$. If every $t$-subset of points occurs in exactly one block, we call it a {\it $t$-$(n,k,1)$ design}, or a $t$-design in short. When $t=2$ and $k=3$, such a $2$-design is called a {\it Steiner  triple system} of order $n$, denoted by STS$(n)$. The {\it leave graph} of a $t$-$(n,k,1)$ packing is the set system $(X,E)$ where $E$ consists of all $t$-subsets of $X$ that do not appear in any blocks. For $t=2$ and $k=3,4$, or $t=3$ and $k=4$, the packing numbers have been completely determined, see \cite{stinson2006packings}, and the corresponding leave graphs are also characterized. We list them below for later use.


\begin{lemma}\cite{stinson2006packings}\label{3packing}
  For any positive $n\not\equiv 5\pmod{6}$, $D(n,3,2)=\lfloor \frac{n}{3} \lfloor \frac{n-1}{2}\rfloor\rfloor$; when $n\equiv 5\pmod{6}$,  $D(n,3,2)=\lfloor \frac{n}{3} \lfloor \frac{n-1}{2}\rfloor\rfloor-1$, and the leave graph consists of a cycle of length four and $n-4$ isolated vertices.

\end{lemma}

\begin{lemma}\cite{stinson2006packings}\label{4packing}
For positive integers $n\not\in\{8,9,10,11,17,19\}$, we have

\begin{equation*}
    D(n,4,2)=\left\{
    \begin{array}{ll}
      \left\lfloor\frac{n}{4}\left\lfloor\frac{n-1}{3}\right\rfloor\right\rfloor-1,& n\equiv 7,10\pmod{12}\vspace{0.1cm},\\
      \left\lfloor\frac{n}{4}\left\lfloor\frac{n-1}{3}\right\rfloor\right\rfloor,& \text{otherwise}.\\
    \end{array}
    \right.
\end{equation*}
When $n=8,9,10,11,17,19$, the values of $D(n,4,2)$  are equal to $2, 3, 5, 6, 20, 25$, respectively.
\end{lemma}

\begin{lemma}\cite{bao2015completion}\label{D(v,4,3)}
  For any positive $n$, it holds that
\begin{equation*}
    D(n,4,3)=\left\{
    \begin{array}{lr}
      \left\lfloor\frac{n}{4}\left\lfloor\frac{n-1}{3}\left\lfloor\frac{n-2}{2}\right\rfloor\right\rfloor\right\rfloor,& \text{if }n\not\equiv 0\pmod{6}\vspace{0.1cm},\\
      \left\lfloor\frac{n}{4}\left(\left\lfloor\frac{n-1}{3}\left\lfloor\frac{n-2}{2}\right\rfloor\right\rfloor-1\right)\right\rfloor,& \text{if } n\equiv 0\pmod{6}.
    \end{array}
    \right.
\end{equation*}
\end{lemma}

 A {\it group divisible design} ($K$-GDD) of order $n$ is a triple $(X,\G,\B)$, where $(X,\B)$ is a set system of order $n$, and $\G$ is a partition of $X$ into subsets, called {\it groups}, such that (1) if $B\in\B$ then $|B|\in K$, (2) every pair of elements of $X$ not contained in a group appears in exactly one block, and (3) every pair contained in a group does not appear in any block. The {\it type} of a GDD is defined as the {\it exponential notation} $g_1^{a_1}g_2^{a_2}\cdots g_s^{a_s}$ if there are $a_i$ groups of size $g_i,i=1,2,\ldots,s$. A $K$-GDD of type $1^n$ is known as a {\it pairwise balanced design} of order $n$, denoted by an $(n,K)$-PBD. In fact, a $K$-GDD is a $2$-$(n,K,1)$ packing, and an $(n,K)$-PBD is a $2$-$(n,K,1)$ packing where every pair of $X$ occurs in exactly one block.

\begin{lemma}\cite{ge2007group}\label{3-GDD}
  The necessary and sufficient conditions for the existence of a 3-GDD of type $3^u$ with $u\geq 3$ is $u\equiv 1\pmod{2}$.
\end{lemma}

\begin{lemma}\cite{wei2014group}\label{GDD}
  Let $u\geq 4$ and $m> 0$. For each $g\in\{2$, $6$, $7$, $9$, $12$, $15$, $24$, $27$, $36\}$, there exists a $4$-GDD of type $g^um^1$ if and only if $m\leq g(u-1)/2,gu\equiv 0\pmod{3},g(u-1)+m\equiv 0 \pmod{3}$ and $\binom{gu+m}{2}-u\binom{g}{2}-\binom{m}{2}\equiv 0\pmod{6}$ except for $(g,u,m)=(2,6,5)$, and except possibly for $(g,u,m)\in\{(2,33,23), (2,33,29), (2,39,35), (6,13,27), (6,13,33), (6,17,39), (6,19,45), (6,19,51), (6,23,63)\}$, and $m\geq0$ if  $g\in\{6$, $12$, $24$, $36\}$.
\end{lemma}

For convenience of later use, we derive the existence of several classes  of explicit GDDs from Lemma~\ref{GDD} as follows, except for the $4$-GDDs of type $7^4$, $9^4$, $9^5$,  $27^4$ and $39^46^1$, in which case they are given in \cite{ge2007group}.

\begin{corollary}\label{coro_4GDD}
  \begin{itemize}
    \item[(1)] There exists a $4$-GDD of type $2^um^1$ for each $u\geq 6, u\equiv 0\pmod{3}$ and $m\equiv 2\pmod{3}$ with $2\leq m\leq u-1$ except for $(u,m)=(6,5)$, and except possibly for $(u,m)\in\{(33,23),(33,29),(39,35)\}$.
    \item[(2)] A $4$-GDD of type $12^um^1$ exists if and only if either $u=3$ and $m=12$, or $u\geq 4$ and $m\equiv 0\pmod{3}$ with $0\leq m\leq 6(u-1)$.
    \item[(3)] A $4$-GDD of type $15^um^1$ exists if and only if either $u\equiv 0\pmod{4}$ and $m\equiv 0\pmod{3}$, $0\leq m\leq (15u-18)/2$; or $u\equiv 1\pmod{4}$ and $m\equiv 0\pmod{6}$, $0\leq m\leq (15u-15)/2$; or $u\equiv 3\pmod{4}$ and $m\equiv 3\pmod{6}$, $0< m\leq (15u-15)/2$.
    \item[(4)] For $u\geq 4$, a $4$-GDD of type $24^um^1$ exists if and only if $m\equiv 0\pmod{3}$ and $0\leq m\leq 12(u-1)$.
    \item[(5)] For $u\geq 4$, a $4$-GDD of type $36^um^1$ exists if and only if $m\equiv 0\pmod{3}$ and $0\leq m\leq 18(u-1)$.
    \item[(6)] There exists $4$-GDDs of types $6^7$, $6^{15}$, $6^{11}30^1$, $6^{12}30^1$, $7^4$, $7^{12}10^1$, $9^4$, $9^46^1$, $9^5$, $9^56^1$, $27^4$, $27^49^1$, $27^5$, $39^46^1$.
  \end{itemize}
\end{corollary}

In the sequel of this paper, all mentioned $4$-GDDs can be found in Corollary~\ref{coro_4GDD} if not specified.

\subsection{The connections between $\ell_1$-metric codes and packings}\label{sec:iic}

We first consider the connection for a binary code, i.e., a set of vectors in $I_2^{n}$. There is a canonical one-to-one correspondence between all vectors $\textbf{u}$ in $I_2^n$ and all subsets $supp(\textbf{u})$  of $[n]$, so a binary code $\C\subset I_2^n$ corresponds to a set system $(X,\{supp(\textbf{u}):\textbf{u}\in \C\})$, where $X=[n]$. For two distinct codewords $\textbf{u},\textbf{v}\in I_2^n$, their distance  equals  $|supp(\textbf{u})\vartriangle supp(\textbf{v})|$. Suppose that $\C$ is an $(n,2t,w)_2$ code, then any two codewords $\textbf{u},\textbf{v}$ satisfy $d(\textbf{u},\textbf{v})\geq 2t$, i.e., $|supp(\textbf{u})\vartriangle supp(\textbf{v})|\geq 2t$, which means their intersection is at most $w-t$.  This is to say that any $(w-t+1)$-subset of $X$ occurs in at most one block of $\{supp(\textbf{u}):\textbf{u}\in \C\}$,  thus $(X,\{supp(\textbf{u}):\textbf{u}\in \C\})$ is a $(w-t+1)$-$(n,w,1)$ packing.

Next we  describe the similar connection for  $q$-ary codes. Note that the distance between any two $q$-ary codewords $\textbf{u}$ and $\textbf{v}$ can be expressed as follows,
\begin{equation}\label{eqdis}
d(\textbf{u},\textbf{v})=2w-2\times\sum_{x\in supp(\textbf{u})\cap supp(\textbf{v})}\min\{\textbf{u}_x,\textbf{v}_x\}.
\end{equation}
When $q=2$, the right hand side reduces to   $|supp(\textbf{u})\vartriangle supp(\textbf{v})|$. Similar to the binary case, to get the same distance $2t$, it is necessary that the intersection of $supp(\textbf{u})$ and $supp(\textbf{v})$ can not exceed $w-t$. So, we have a universal necessary condition, which we state below and call it the \textbf{UNC} condition.

 \vspace{0.2cm}
 \emph{\textbf{UNC}: If $\C$ is an $(n,2t,w)_q$ code, then the set system $(X,\{supp(\textbf{u}):\textbf{u}\in \C, |supp(\textbf{u})|\geq \tau\})$ is a $\tau$-$(n,\{w,w-1,\ldots,\tau\},1)$ packing, where $\tau \triangleq w-t+1$.}\vspace{0.2cm}\\
 Note that the supports may be the same for two distinct codewords, but this can only occur when both supports are of size less than $w-t+1$, which we do not include in the packing.

 Further, by Eq (\ref{eqdis}), the minimum value between  nonzero entries $\textbf{u}_x$ and $\textbf{v}_x$ can not be big if the intersection size, $|supp(\textbf{u})\cap supp(\textbf{v})|$ closes to $w-t$. For example, if the intersection size is exactly $w-t$, then each minimum value in the summation must be one. If the other extremal case happens, that is, the intersection size is just one, then the minimum value could be any value upper bounded by $w-t$. This is also a key observation when we deduce the necessary conditions for a code with certain distance.
 To indicate different nonzero entries, we associate a $q$-ary codeword $\textbf{u}$ with a subset $\phi(\textbf{u}):=\{(x,i):x\in\text{supp}(\textbf{u})\text{ and }\textbf{u}_x=i\}\subset X\times [q-1]$, where $X=[n]$. To simplify the notation, we will write $x_i$ in $\phi(\textbf{u})$ instead of $(x,i)$. It is difficult to characterize a general condition of the set system $(X\times [q-1],\{\phi(\textbf{u}):\textbf{u}\in \C\})$ to be an $(n,2t,w)_q$ code, since codes with different $q$, $t$ or $w$ may result in very different requirements.
By abuse of notation we sometimes do not distinguish between $\textbf{u}$ and $\phi(\textbf{u})$, since they refer to equivalent objects. For example, we have a $(6,6,4)_3$  code $\C\subset I_3^n$ with four codewords $210100$, $021010$, $002101$ and $100012$. Equivalently, we can describe them as $\{1_2,2_1,4_1\}$, $\{2_2,3_1,5_1\}$, $\{3_2,4_1,6_1\}$, $\{6_2,5_1,1_1\}$, which are subsets of $[6]\times [2]$. In the sequel of this paper, besides the conventional set $[n]$, we also assume that $X$ is an ordered set for some $X\neq [n]$.

%

When constructing certain codes in the following sections, we usually employ some group action on the set of positions of the code. So it is convenient for us to use some additive group like $\bbZ_n$ instead of the conventional set $[n]$, and give elements of $\bbZ_n$ a natural order. For example,   we have a $(4,4,3)_3$  code $\C\subset I_3^4$ with four codewords $1200$, $0120$, $0012$ and $2001$, or equivalently, we have a $(4,4,3)_3$  code $\C\subset \bbZ_4 \times [2]$ with codewords $\{0_1,1_2\}$, $\{1_1,2_2\}$, $\{2_1,3_2\}$, $\{3_1,0_2\}$. By the set notations, we see that $\C=\{\{(0+i)_1,(1+i)_2\}:i\in\bbZ_4\}$, that is,  $\C$ is obtained from a codeword $\{0_1,1_2\}$ by a group action $\bbZ_4$ on the set of positions.
We call such a codeword $\{0_1,1_2\}$ a {\it base codeword (or a base block)}. Motivated by this example, when we try to find certain code, we only need to find a base codeword (or a set of base codewords) with extra conditions so that it could be developed to the desired code.  This is a very common method in combinatorial design theory to find small designs by computer search, and will be very helpful for us to find small codes. To save space, when we write down the code, we only list the base codewords and the employed group action instead of listing all codewords. Note that the group action employed usually acts on the set of positions of the codewords, but leaves the nonzero entries fixed.

%
%

\section{CWCs  over Non-negative Integers}\label{cwcz}

In this section, we consider codes with constant weight $w\leq 4$ over $\bbZ_{\geq 0}$ and determine the exact value of $A(n,d,w)$ for $w=3,4$ and all distances $d=2t$ between  $4$ and $2w-2$. We can choose a small integer $q\geq w+1$, then the code of constant weight $w$ over $\bbZ_{\geq 0}$ can be viewed as a $q$-ary code. As we observed in Section~\ref{sec:iic}, the set system $(X,\{supp(\textbf{u}):\textbf{u}\in \C, |supp(\textbf{u})|\geq \tau\})$ is a $\tau$-$(|X|,\{w,w-1,\ldots,\tau\},1)$ packing, and  the minimum value between nonzero entries in the same position can not be big if the intersection size closes to $w-t$.

For weight three, there are codewords of three different types  $1^3$, $1^12^1$, and  $3^1$. Since $w=3$, we only need to consider $d=4$, then $t=2$ and $\tau=2$. The following Lemma is obtained by observing the limitation of distance.

\begin{lemma}\label{lem3_1}
  A code $\C\subset \bbZ_n\times \bbZ_{\geq 0}$ consisting of codewords of types $1^3,1^12^1$ and $3^1$ is an $(n,4,3)$ code if and only if the following conditions are satisfied:
\begin{itemize}
  \item[(1)] The collection of subsets $supp(\textbf{u})\subset \bbZ_n$, for all codewords $\textbf{u}$  of types $1^3$ or $1^12^1$, forms a $2$-$(n,\{2,3\},1)$ packing.
  \item[(2)] For any two codewords $\textbf{u},\textbf{v}$, if $i\in supp(\textbf{u})\cap supp(\textbf{v})$, then $\min\{\textbf{u}_i,\textbf{v}_i\}=1$.
\end{itemize}
\end{lemma}

\begin{proof} The necessity of (1) is from the \textbf{UNC} condition, and of (2) is by the the distance formula in Eq. (\ref{eqdis}). The sufficiency is from direct verification. In fact, by (1), any two codewords have at most one common position in their supports. If two codewords have disjoint supports, then the distance is six. If they have one common position in supports, then by Eq. (\ref{eqdis}) and condition (2), the distance is four since the minimum value is $1$.
\end{proof}

By Lemma~\ref{lem3_1}, we can deduce an upper bound of $A(n,4,3)$ and construct an $(n,4,3)$ code achieving this bound by packings.

\begin{theorem}
  $A(n,4,3)=D(n,3,2)+n.$
\end{theorem}

\begin{proof}
  Let $x,y$ and $z$ be the number of codewords  of types $1^3,1^12^1$ and $3^1$, respectively. By condition (1) of Lemma~\ref{lem3_1}, we have \[x\leq D(n,3,2),\] since codewords of type $1^3$ form a $2$-$(n,3,1)$ packing. By  condition (2) of Lemma~\ref{lem3_1}, there does not exist a pair of codewords $\textbf{u},\textbf{v}$ with  $\textbf{u}_i\geq 2$ and $\textbf{v}_i\geq 2$ simultaneously for some position $i$. Thus we have \[y+z\leq n\]by counting the occurrences of symbols $2$ and $3$ in all codewords. So the upper bound $A(n,4,3)=x+y+z\leq D(n,3,2)+n$ follows.
The lower bound is achieved by the code consisting of all binary codewords of type $1^3$ obtained from an optimal $2$-$(n,3,1)$ packing, and $n$ codewords of type $3^1$ with disjoint supports.
\end{proof}

For weight four, there are codewords of five different types  $1^4$, $1^22^1$, $1^13^1$, $2^2$, and $4^1$. We first consider distance $d=4$, that is $t=2$ and $\tau=3$. Similar to Lemma~\ref{lem3_1}, a code $\C\subset \bbZ_n\times \bbZ_{\geq 0}$ consisting of codewords of constant weight four is an $(n,4,4)$ code if and only if the following conditions are satisfied:
\begin{itemize}
  \item[(i)] The collection of subsets $supp(\textbf{u})\subset \bbZ_n$, for all codewords $\textbf{u}$  of types $1^4$ or $1^22^1$, forms a $3$-$(n,\{3,4\},1)$ packing.
  \item[(ii)] For any two codewords $\textbf{u},\textbf{v}$, if $i\in supp(\textbf{u})\cap supp(\textbf{v})$, then $\min\{\textbf{u}_i,\textbf{v}_i\}\leq 2$.
       \item[(iii)] For any two codewords $\textbf{u},\textbf{v}$, if  $\{i, i'\} \subset supp(\textbf{u})\cap supp(\textbf{v})$, then $\min\{\textbf{u}_i,\textbf{v}_i\}=\min\{\textbf{u}_{i'},\textbf{v}_{i'}\}=1$.
\end{itemize}
The necessary condition (i) is from the \textbf{UNC} condition, and (ii) and (iii) are from Eq.~(\ref{eqdis}). That is, if $|supp(\textbf{u})\cap supp(\textbf{v})|=1$, then the minimum value could be at most $2$. But if $|supp(\textbf{u})\cap supp(\textbf{v})|=2$, then both minimum values must be $1$. The intersection of supports for all pairs of codewords can not exceed two by Eq.~(\ref{eqdis}), and equivalently by condition (i).

By the above necessary and sufficient conditions, we can determine the size of the largest $(n,4,4)$ code.


\begin{theorem}\label{thmz1}
  $A(n,4,4)=D(n,4,3)+\frac{n(n-1)}{2}+n.$
\end{theorem}

\begin{proof}  Let $x,y,z,a$ and $b$ be the number of codewords  of types $1^4$, $1^22^1$, $2^2$, $1^13^1$ and $4^1$,  respectively.
  By conditions (i) and (ii), we have
  \[ x\leq D(n,4,3), \text{ and }a+b\leq n.\]
From condition (iii),  all ordered pairs $(i,i')$ satisfying $\textbf{u}_i\geq 1$ and $\textbf{u}_{i'}\geq 2$ for some codeword $\textbf{u}\in \C$  should be different. By counting such pairs, we have
  \begin{equation*}
    2y+2z+a\leq n(n-1).
  \end{equation*}
  Combining the above inequalities, we have that $A(n,4,4)=x+y+z+a+b\leq D(n,4,3)+\frac{n(n-1)}{2}+n$. Note that  every codeword of type $4^1$ or $2^2$ has distance at least four with any other codewords whose type are neither $1^22^1$ nor $1^13^1$, and this gives $\binom{n}{2}+n$ codewords. So the lower bound is achieved by the code consisting of all binary codewords of type $1^4$ obtained from an optimal 3-$(n,4,1)$ packing over $\bbZ_n$, and all different codewords of types $2^2$ and $4^1$.
%
\end{proof}

For  distance $d=6$, we have $t=3$ and $\tau=2$. By the \textbf{UNC} condition  and the the distance formula in Eq. (\ref{eqdis}),  a code $\C\subset \bbZ_n\times \bbZ_{\geq 0}$ of constant weight four is an $(n,6,4)$ code if and only if the following conditions are satisfied:
\begin{itemize}
  \item[(a)] The collection of subsets $supp(\textbf{u})\subset \bbZ_n$, for all codewords $\textbf{u}$  of types $1^4$, $1^22^1$,  $1^13^1$, or $2^2$, forms a $2$-$(n,\{2,3,4\},1)$ packing.
   \item[(b)]For any two codewords $\textbf{u},\textbf{v}$, if $i\in supp(\textbf{u})\cap supp(\textbf{v})$, then $\min\{\textbf{u}_i,\textbf{v}_i\}=1$.
\end{itemize}


\begin{theorem}\label{thmz2}
  $A(n,6,4)=D(n,4,2)+n.$
\end{theorem}

\begin{proof} Using the same notation as in the proof of Theorem~\ref{thmz1}, we have
  \begin{align*}
    & x\leq D(n,4,2), \text{ and }
    y+2z+a+b\leq n.
  \end{align*}  by conditions (a) and (b). Then $A(n,6,4)=x+y+z+a+b\leq D(n,4,2)+n$. The lower bound is achieved by the code consisting of all binary codewords of type $1^4$ obtained from an optimal 2-$(n,4,1)$ packing over $\bbZ_n$ and $n$ codewords of type $4^1$.\end{proof}

In the sequel of this paper, we only consider ternary constant-weight codes, and the basic idea is also using the optimal packings to construct optimal codes.

\section{Ternary CWCs of Weight Three}\label{cwc3}

In this section, we consider ternary codes with constant weight  three in $I_3^n$ and determine the exact value of $A_3(n,4,3)$.
Since the code is ternary, there are codewords of two different types $1^3$ and $1^12^1$.  For  distance $d=4$, we have $t=2$ and $\tau=2$. By the \textbf{UNC} condition  and the the distance formula in Eq. (\ref{eqdis}), a code $\C\subset \bbZ_n\times [2]$ of constant weight three  is an $(n,4,3)_3$ code if and only if the following conditions are satisfied:
\begin{itemize}
  \item[$(1')$] The collection of subsets $supp(\textbf{u})\subset \bbZ_n$, for all codewords $\textbf{u}\in\C$, forms a $2$-$(n,\{2,3\},1)$ packing.
  \item[$(2')$]  For any two codewords $\textbf{u},\textbf{v}$, if $i\in supp(\textbf{u})\cap supp(\textbf{v})$, then $\min\{\textbf{u}_i,\textbf{v}_i\}=1$.
\end{itemize}

\begin{lemma}\label{thm1}
  $A_3(n,4,3)\leq\left\lfloor\frac{n^2+3n}{6}\right\rfloor$.
\end{lemma}

\begin{proof}
  Let $x$ and $y$ be the number of codewords of types $1^3$ and $1^12^1$ respectively. By conditions $(1')$ and $(2')$, we have
  \begin{align*}
    & 3x+y\leq \binom{n}{2},\text{ and } y\leq n.
  \end{align*}
  The former inequality is by the definition of packings, that is, the number of pairs contained in all blocks can not exceed $\binom{n}{2}$. The latter one is by counting the occurrences of symbol $2$. Then $A_3(n,4,3)=x+y\leq\left\lfloor\frac{n^2+3n}{6}\right\rfloor$.\end{proof}



By the proof of Lemma~\ref{thm1}, it is possible for a code to achieve the upper bound when $y=n$ or $n-1$. Assume we have already found $n$ or $n-1$ codewords of type $1^12^1$ satisfying $(2')$,  then we need to find a $2$-$(n,3,1)$ packing such that the condition $(1')$ is satisfied. If the size of this $2$-$(n,3,1)$ packing is $\left\lfloor\frac{n^2+3n}{6}\right\rfloor-n$ or $\left\lfloor\frac{n^2+3n}{6}\right\rfloor-n+1$, respectively, then the upper bound in Lemma~\ref{thm1} can be achieved.



  \begin{theorem}\label{th1}
    $A_3(n,4,3)=\left\lfloor\frac{n^2+3n}{6}\right\rfloor$.
  \end{theorem}
  \begin{proof}The upper bound is from Lemma~\ref{thm1}. For the lower bound, the case $n\leq 3$ is easy to check; for all other integers $n$, we construct an $(n,4,3)_3$ code $\C$ of size achieving the upper bound as follows.

    For each $n\equiv 1,5\pmod{6}$, there exists a $2$-$(n,3,1)$ packing $(X,\B)$ of size $\frac{n^2-3n+2}{6}$  whose leave graph consists of a cycle  of length $n-1$ and one isolated vertex \cite{colbourn1986quadratic}. We assume that $X=\bbZ_{n-1}\cup\{\infty\}$, and  the cycle is $(0,1,2,\ldots,n-2)$. Then the desired $\C$ consists of all codewords $\{a_1,b_1,c_1\}$ of type $1^3$ for each $\{a,b,c\}\in \B$, and $(n-1)$ codewords of type $1^12^1$: $\{0_1,1_2\},\{1_1,2_2\},\ldots,\{(n-3)_1,(n-2)_2\},\{(n-2)_1,0_2\}$.

    For each $n\equiv 2,4\pmod{6}$, there exists a $2$-$(n-1,3,1)$ design $(\bbZ_{n-1},\B)$ of size $\frac{(n-1)(n-2)}{6}$, see Lemma~\ref{3packing}, which is indeed an  STS$(n-1)$. The desired code is constructed over $\bbZ_{n-1}\cup \{\infty\}$, which consists of  all codewords $\{a_1,b_1,c_1\}$ of type $1^3$ for each $\{a,b,c\}\in \B$, and $(n-1)$  codewords of type $1^12^1$: $\{0_2,\infty_1\},\{1_2,\infty_1\},\ldots,\{(n-3)_2,\infty_1\},\{(n-2)_2,\infty_1\}$.

      When $n\equiv 3\pmod{6}$, Colbourn and Rosa\cite{colbourn1986quadratic} (and Colbourn and Ling\cite{colbourn1998class}) showed that there exists a $2$-$(n,3,1)$ packing $(\bbZ_{n},\B)$ of size $\frac{n^2-3n}{6}$  whose leave graph consists of all pairs $\{i,j\}$ with $i-j\equiv \pm 1\pmod{n}$. Then the desired $\C$ consists of all codewords $\{a_1,b_1,c_1\}$ of type $1^3$ for each $\{a,b,c\}\in \B$, and $n$  codewords of type $1^12^1$:  $\{0_1,1_2\},\{1_1,2_2\},\ldots,\{(n-2)_1,(n-1)_2\},\{(n-1)_1,0_2\}$.

      When $n\equiv 0\pmod{6}$, there is a $2$-$(n-1,3,1)$ packing $(\bbZ_{n-1},\B)$  of size $\frac{n^2-3n-6}{6}$ whose leave graph consists of a cycle of length four and $n-5$ isolated vertices by Lemma~\ref{3packing}. Assume the cycle is $(0,1,2,3)$.  The desired code $\C$ is constructed over $\bbZ_{n-1}\cup \{\infty\}$,  which consists of all codewords $\{a_1,b_1,c_1\}$ of type $1^3$ for each $\{a,b,c\}\in \B$, an additional codeword $\{1_1,2_1,\infty_1\}$ of  type $1^3$, five codewords of type $1^12^1$ constructed from the cycle: $\{0_2,\infty_1\}$, $\{0_1,1_2\}$, $\{2_2,3_1\}$, $\{0_1,3_2\}$, $\{3_1,\infty_2\}$, and $(n-5)$  codewords of type $1^12^1$:  $\{4_2,\infty_1\},\{5_2,\infty_1\},\ldots,\{(n-3)_2,\infty_1\},\{(n-2)_2,\infty_1\}$.

  It is routine to check that all codes constructed above are $(n,4,3)_3$ codes of the required sizes.
  \end{proof}
%
%

Next, we give examples to illustrate constructions in Theorem~\ref{th1}.
  \begin{example} For $n=6$, $A_3(6,4,3)=9.$ There exists a $2$-$(5,3,1)$ packing over $\bbZ_5$ with blocks $024,134$ whose leave graph is a cycle $(0,1,2,3)$. We adjoin an infinity point $\infty$, and construct an optimal code with the following codewords:
    \[\begin{array}{lll}
    \{0_1,2_1,4_1\}&\{1_1,3_1,4_1\}&\{1_1,2_1,\infty_1\}\\
    \{0_2,\infty_1\}&\{0_1,1_2\}&\{2_2,3_1\}\\
    \{0_1,3_2\}&\{3_1,\infty_2\}&\{4_2,\infty_1\}.
      \end{array}
    \]

For $n=7$, we construct an optimal $(7,4,3)_3$ code of size eleven as follows.
     Since there is a $2$-$(7,3,1)$ packing with five blocks $14\infty,25\infty,03\infty,135,024$ over $\bbZ_6\cup\{\infty\}$ whose leave graph consists of a cycle $(0,1,2,3,4,5)$, then the codewords are as follows:
    \[\begin{array}{llllllll}
    \{1_1,4_1,\infty_1\}&\{2_1,5_1,\infty_1\}&\{0_1,3_1,\infty_1\}&\{1_1,3_1,5_1\}\\
    \{0_1,2_1,4_1\}&\{0_1,1_2\}&\{1_1,2_2\}&\{2_1,3_2\}\\
    \{3_1,4_2\}&\{4_1,5_2\}&\{5_1,0_2\}.
      \end{array}
    \]

  For $n=8$, we construct an optimal $(8,4,3)_3$ code of size fourteen as follows.
     From an STS$(7)$ over $\bbZ_7$ with blocks $124,235,346,045,156,026,013$, adjoining an infinite point $\infty$, we obtain the codewords as follows:
    \[\begin{array}{lllllll}
    \{1_1,2_1,4_1\}&\{2_1,3_1,5_1\}&\{3_1,4_1,6_1\}&\{0_1,4_1,5_1\}\\
    \{1_1,5_1,6_1\}&\{0_1,2_1,6_1\}&\{0_1,1_1,3_1\}&\{0_2,\infty_1\}\\
    \{1_2,\infty_1\}&\{2_2,\infty_1\}&\{3_2,\infty_1\}&\{4_2,\infty_1\}\\
    \{5_2,\infty_1\}&\{6_2,\infty_1\}.
  \end{array}
\]

  For $n=9$,  $A_3(9,4,3)=18$. The blocks generated by $035$ under $\bbZ_9$ is a $2$-$(9,3,1)$ packing, whose leave graph is a cycle of length nine. Then the codewords are as follows:
    \[\begin{array}{llll}
    \{0_1,3_1,5_1\}&\{1_1,4_1,6_1\}&\{2_1,5_1,7_1\}&\{3_1,6_1,8_1\}\\
    \{4_1,7_1,0_1\}&\{5_1,8_1,1_1\}&\{6_1,0_1,2_1\}&\{7_1,1_1,3_1\}\\
    \{8_1,2_1,4_1\}&\{0_1,1_2\}&\{1_1,2_2\}&\{2_1,3_2\}\\
    \{3_1,4_2\}&\{4_1,5_2\}&\{5_1,6_2\}&\{6_1,7_2\}\\
    \{7_1,8_2\}&\{8_1,0_2\}.
      \end{array}
    \]

  \end{example}

\begin{remark} \label{gddconst} For $n\equiv 3\pmod{6}$, we give another construction of optimal codes as follows. Let $u=n/3\geq 3$, then there exists a  $3$-GDD of type $3^u$ by Lemma~\ref{3-GDD}, say $(X,\G, \B)$.  For each group $G=\{a,b,c\}\in \G$, we obtain three codewords of type $1^12^1$: $\{a_1,b_2\}$, $\{b_1,c_2\}$ and $\{c_1,a_2\}$. It is easy to check that all these $n$ codewords of type  $1^12^1$, combining all codewords of type $1^3$ obtained from $\B$, form an optimal $(n,4,3)_3$ code.

In fact, the above $3$-GDD of type $3^u$ is also a $2$-$(n,3,1)$ packing of size $\frac{n^2-3n}{6}$, whose leave graph is a union of $u$ cycles $(a,b,c)$ of length three, where $\{a,b,c\}$ is a group in $\G$. Together with the $2$-$(n,3,1)$ packing given in Theorem~\ref{th1} for $n\equiv 3\pmod{6}$, we simply have two non-isomorphic packings of the same size, but with different leave graphs, and we can use either to construct optimal codes.

However, the leave graph of a GDD is a union of disjoint cliques,  which is convenient for us since we can input an optimal short code over each clique and do not destroy the packing property. Note that in the above example, we input an optimal $(3,4,3)_3$ code over each group of size three.  This observation is very helpful, since there are plenty results of existence of GDDs as in Lemma~\ref{GDD}, and we only need to find optimal short codes over each group. This method will be used to construct optimal $(n,6,4)_3$ codes by using $4$-GDDs in the next section.

\end{remark}

\section{Ternary CWCs of Weight Four}\label{conca}

In this section, we consider ternary codes with constant weight  four and determine the exact value of $A_3(n,d,4)$ for $d=4$ and $6$.  The ternary codes with constant weight four have codewords of three different types $1^4$, $1^22^1$, and $2^2$.

When $d=4$, the only difference between the CWC of weight four over $\bbZ_{\geq 0}$ and $I_3$ is that the latter one does not have codewords of types $1^13^1$ and $4^1$. Then by similar arguments as in the proof of Theorem~\ref{thmz1}, we can obtain the following result.

\begin{theorem}\label{thm_A_3(n,4,4)}
 $A_3(n,4,4)=D(n,4,3)+\frac{n(n-1)}{2}$.
 \end{theorem}

 Now we consider $d=6$. By the \textbf{UNC} condition and the  distance formula in Eq. (\ref{eqdis}),  a ternary code $\C$ of constant weight four is an $(n,6,4)_3$ code if and only if $\C$ satisfies the following properties.
\begin{itemize}
  \item[$(a')$] The collection of subsets $supp(\textbf{u})\subset \bbZ_n$, for all codewords $\textbf{u}\in\C$, forms a $2$-$(n,\{2,3,4\},1)$ packing.
  \item[$(b')$]  For any two codewords $\textbf{u},\textbf{v}\in \C$, if $i\in supp(\textbf{u})\cap supp(\textbf{v})$, then $\min\{\textbf{u}_i,\textbf{v}_i\}=1$.
\end{itemize}

From these properties, we have the following upper bound for $A_3(n,6,4)$.
\begin{lemma}\label{A3(n,6,4)}
  $A_3(n,6,4)\leq \left\lfloor\frac{n(n+5)}{12}\right\rfloor=:U(n)$.
\end{lemma}
\begin{proof}
Let $x, y, z$ be the number of codewords of types $1^4$, $1^22^1$ and $2^2$,  respectively. Property $(a')$ implies that the number of $2$-subsets of $\bbZ_n$ contained in the support of all codewords cannot exceed $\binom{n}{2}$, by the definition of a $2$-$(n,\{2,3,4\},1)$ packing. So
  \begin{align*}
    & 6x+3y+z\leq \binom{n}{2}.
  \end{align*}
   Property $(b')$ indicates that in any position, the symbol $2$ can only appear in at most one codeword, that is,
    \begin{align*}
   y+2z\leq n.
  \end{align*}
Then we get  $A_3(n,6,4)=x+y+z\leq \left\lfloor\frac{n(n+5)}{12}\right\rfloor$.
\end{proof}

 \begin{remark}\label{case_z}
  If $A_3(n,6,4)=U(n)$, then the number $z$ of codewords of type $2^2$ satisfies  $z=0$ when $n\equiv 0,3,4,7\pmod{12}$, $z\leq 1$ when $n\equiv 2,5\pmod{12}$, $z\leq 3$ when $n\equiv 1,6,9,10\pmod{12}$ and $z\leq 4$ when $n\equiv 8,11\pmod{12}$.

  In fact, in the proof of Lemma~\ref{A3(n,6,4)}, we indeed obtain the following inequality\[x+y+z\leq\left\lfloor\frac{n(n+5)-2z}{12}\right\rfloor.\]
When $n\equiv 1,6,9,10\pmod{12}$, we have $U(n)=\frac{n(n+5)-6}{12}$. If $z>3$, then $x+y+z\leq U(n)-1$, a contradiction. So we have $z\leq 3$ when $n\equiv 1,6,9,10\pmod{12}$. The other cases are similar.
\end{remark}

The rest of this section serves to construct $(n,6,4)_3$ codes that achieve the upper bound. From the insight described in  Remark~\ref{gddconst}, we will make use of the known existence of GDDs, and input optimal codes of short length to construct optimal codes of long length. We first handle the existence of codes with short length.

\subsection{Direct constructions for optimal codes of short length}
It is easy to see that  $A_3(1,6,4)=0=U(1)$, $A_3(2,6,4)=1=U(2)$, $A_3(3,6,4)=1=U(3)-1$ and $A_3(4,6,4)=2=U(4)-1$. For $n=5$, there is no $(5,6,4)_3$ code of size four by exhaustive search. Then $A_3(5,6,4)=3=U(5)-1$ since $21100$, $10012$, $02020$ form a $(5,6,4)_3$ code. For $n\in [6,11]$, we give direct constructions of codes of size $U(n)$.


\begin{lemma}\label{small_base}
For all $n\in [6,11]$, we have $A_3(n,6,4)=U(n)$.
\end{lemma}
\begin{proof} For each $n\in [6,11]$, we construct an optimal code as follows.
%

  For $n=6$, $A_3(6,6,4)= 5$. From an STS$(7)$ over $\bbZ_7$ with blocks $124,235,346,450,561,602,013$, delete the point $6$ and all blocks containing it. Use the remaining blocks to get four codewords of type $1^22^1$ by choosing different positions $\{0,1,2,5\}$ for the symbol $2$. Then add a codeword of type $2^2$ by using the pair $\{3,4\}$. The final code is listed as follows.
  \[\begin{array}{lll}
    \{0_2,1_1,3_1\}&\{1_2,2_1,4_1\}&\{2_2,3_1,5_1\}\\
    \{4_1,5_2,0_1\}&\{3_2,4_2\}.
      \end{array}
  \]

  For $n=7$, $A_3(7,6,4)=7$.  We use the same STS$(7)$ which indeed has a base block $013$ under the group action $\bbZ_7$. Then the codewords of our code $\C$  are generated by $\{0_2,1_1,3_1\}$ under the group action $\bbZ_7$ on the set of coordinates.

  For $n=8$, $A_3(8,6,4)= 8$. The optimal code $\C$ is obtained by developing the codeword $\{0_2,1_1,3_1\}$ under the group action $\bbZ_8$.

  For $n=9$, $A_3(9,6,4)=10$. According to Remark~\ref{case_z} and Property $(b')$, the number of codewords of type $2^2$ is at most $3$, and symbol $2$ can only appear in one codeword at the same position. Then by controlling the number of type $2^2$ codewords,  we have the following code $\C$.
  \[\begin{array}{llll}
    \{0_1,1_1,2_1,3_1\}&\{0_1,4_1,5_1,6_1\}&\{7_1,8_1,0_2\}&\{3_1,6_1,7_2\}\\
    \{5_1,7_1,2_2\}&\{4_1,7_1,1_2\}&\{1_1,8_1,6_2\}&\{3_1,8_1,5_2\}\\
    \{2_1,4_1,8_2\}&\{3_2,4_2\}.
    \end{array}
  \]

  For $n=10$, $A_3(10,6,4)=12$. This code $\C$ is given by a $(10,\{3,4\})$-PBD  by assigning the symbol $2$ appropriately.
  \[\begin{array}{llll}
    \{0_1,1_1,2_1,3_1\}&\{0_1,4_1,5_1,6_1\}&\{0_1,7_1,8_1,9_1\}&\{2_1,4_2,9_1\}\\
    \{3_1,5_1,9_2\}&\{1_1,6_2,9_1\}&\{1_2,4_1,7_1\}&\{2_2,5_1,7_1\}\\
    \{2_1,6_1,8_2\}&\{3_2,4_1,8_1\}&\{1_1,5_2,8_1\}&\{3_1,6_1,7_2\}.
    \end{array}
  \]

  For $n=11$, $A_3(11,6,4)=14$. The code $\C$ is obtained by computer search.
  \[\begin{array}{llll}
    \{0_1,2_1,4_1,6_1\}&\{2_1,3_1,7_1,9_1\}&\{3_1,6_1,8_1,10_1\}&\{0_2,8_1,9_1\}\\
    \{1_1,9_2,10_1\}&\{2_2,5_1,10_1\}&\{4_1,7_1,10_2\}&\{4_2,5_1,9_1\}\\
    \{0_1,3_1,5_2\}&\{1_1,5_1,6_2\}&\{0_1,1_1,7_2\}&\{5_1,7_1,8_2\}\\
    \{1_2,2_1,8_1\}&\{1_1,3_2,4_1\}.
    \end{array}
  \]

\end{proof}

The case of $n=12$ is a little complicated, since the upper bound $U(12)$ cannot be achieved.

\begin{lemma}\label{12}
  $A_3(12,6,4)=U(12)-1=16.$
\end{lemma}

\begin{proof} When $n=12$, $U(n)=17$. By the proof of Lemma~\ref{A3(n,6,4)} and Remark~\ref{case_z}, a $(12,6,4)_3$ code achieves the size $17$ only when $z=0$, $x=5$, $y=12$, and $supp(\textbf{u})\subset \bbZ_{12}$, for all codewords $\textbf{u}$  of types $1^4$ or $1^22^1$, form a $2$-$(12,\{3,4\},1)$ packing. Such a packing has an empty leave graph, since the support of each codeword of type $1^4$ and $1^22^1$ contains six and three $2$-subsets respectively, we have $5\times 6+12\times 3=66=\binom{12}{2}$, i.e., all $2$-subsets of $\bbZ_{12}$ appear in exactly one support of a codeword. Suppose there exists such a code $\C\subset \bbZ_{12}\times [2]$. For each $i\in \bbZ_{12}$, let $x_i$ be the number of codewords of type $1^4$ with a nonzero entry in  position $i$. By property $(a')$, we have $x_i\leq\left\lfloor\frac{11}{3}\right\rfloor=3$ for each $i$ by counting the pairs containing $i$. By counting the nonzero entries in the five codewords of type $1^4$, we have
  \begin{equation*}
    x_0+x_1+\cdots+x_{11}=20.
  \end{equation*}

Since $supp(\textbf{u})\subset \bbZ_{12}$, $\textbf{u}\in \C$ form a $2$-$(12,\{3,4\},1)$ packing with an empty leave graph, the number of pairs containing $i$ is $11=3x_i+2y_i$, where $y_i$ is the number of codewords of type $1^22^1$ having a nonzero entry in position $i$. This forces $x_i$ to be an odd integer, which might be $1$ or $3$.   Let $d_j$  be the number of positions $i$ such that $x_i=j$ for $j=1,3$, then
    \begin{equation*}
    \left\{
    \begin{array}{ll}
      d_1+3d_3=20,\\
      d_1+d_3=12.
    \end{array}
    \right.
\end{equation*}
Therefore, $d_1=8$ and $d_3=4$. Without loss of generality, assume that $x_0=3$, and we have three codewords $\{0_1,1_1,2_1,3_1\}$, $\{0_1,4_1,5_1,6_1\}$, $\{0_1,7_1,8_1,9_1\}$. Since $d_1=8$ and $d_3=4$, there must be a point $i\in [9]$ such that $x_i=3$, we can further assume that $x_1=3$. If $\{10_1,11_1\}$ is contained in a codeword of type $1^4$ containing $1_1$, might set it $\{1_1,4_1,10_1,11_1\}$, in which case we cannot find one more codeword to make $x_1=3$, a contradiction. Thus we can assume we have two more codewords $\{1_1,4_1,7_1,10_1\},\{1_1,5_1,8_1,11_1\}$. To make sure $x_5$ is odd, we need to construct more codewords of type $1^4$, but this will lead to more than five codewords of type $1^4$, a contradiction. So we prove that $A_3(12,6,4)\leq 16$.

In fact, we can construct a $(12,6,4)_3$ code of size $16$ as follows.
\[\begin{array}{llll}
    \{3_1,7_1,8_1,11_1\}&\{1_1,3_1,4_1,6_1\}&\{2_1,6_1,8_1,10_1\}&\{0_1,2_1,7_1,9_1\}\\
    \{4_1,11_1,0_2\}&\{1_1,2_1,11_2\}&\{10_1,11_1,5_2\}&\{0_1,10_1,3_2\}\\
    \{7_1,10_1,4_2\}&\{6_1,11_1,9_2\}&\{1_1,9_1,10_2\}&\{5_1,9_1,8_2\}\\
    \{3_1,5_1,2_2\}&\{1_1,5_1,7_2\}&\{0_1,8_1,1_2\}&\{0_1,5_1,6_2\}.
    \end{array}
  \]
\end{proof}

\subsection{Recursive constructions  based on $4$-GDDs}

In this section, we recursively construct $(n,6,4)_3$ codes of size $U(n)$ by using $4$-GDDs and small optimal codes motivated by Remark~\ref{gddconst}. Our main idea is as follows: given a $4$-GDD of order $n$,  take all its blocks as codewords of type $1^4$ and length $n$ in a natural way; for each group of the GDD, say  of length $g$, take an optimal $(g,6,4)_3$ code and then extend it to a code of length $n$ by assigning zeros to the remaining coordinates; then collect all these codewords. In this method, it is convenient to view a code as a $2$-$(n,\{2,3,4\},1)$ packing by property $(a')$, and then carefully distribute the symbol $2$ so that each position has at most one symbol $2$ among all codewords by property $(b')$.

Although the construction looks simple, it is not easy to guarantee that the resulting long code has size achieving the upper bound $U(n)$, and the length $n$ could run over all positive integers. Our strategy can be described  as follows, which heavily depends on the existence of $4$-GDDs and good small codes.
\begin{itemize}
\item[(1)]\textbf{Good classes of $4$-GDDs}. Choose fixed integers $g$ and $m$, such that the $4$-GDD of type $g^um^1$ exists for all large $u$. So the existence of an optimal short $(g,6,4)_3$  code and an optimal $(m,6,4)_3$ code may result in optimal codes of length $n=gu+m$ for all $u$. One of the choices is that $g\equiv 0\pmod{12}$ and $m\equiv 0\pmod{3}$ by Corollary~\ref{coro_4GDD}. If this method works, then it will lead to optimal codes of length $n$ for all $n\equiv 0\pmod{3}$.
    \item[(2)]\textbf{Adding extra points}. The $4$-GDDs of type $g^um^1$ with $g\equiv 0\pmod{12}$ and $m\equiv 0\pmod{3}$ are good, but only lead to codes with length $n\equiv 0\pmod{3}$. To cover other lengths by using the same $4$-GDDs, we need to add one or two extra points when we input the short codes. That is, for $t=1$ or $2$, we input a suitable $(g+t,6,4)_3$ code or an $(m+t,6,4)_3$ code over each group joined by the same $t$ extra points, and obtain codes of length $n=gu+m+t\equiv t\pmod{3}$.
        \item[(3)]\textbf{Good short codes}. The $(g+t,6,4)_3$ code with $t\geq 0$ should be chosen very carefully since it will be used for each group combined with the same extra points if any. First, when $t=1,2$,  this code can not have symbol $2$ on the extra points by $(b')$. Second,  when $t=2$, the corresponding packing should not contain the pair of the extra two points by $(a')$. Finally, the leave graph of the packing should be empty (except for $t=2$, the leave graph has only one edge joined by two extra points), since otherwise,  the resultant code will have unbounded number of edges in its leave graph when $u$ increases, which can not have size $U(n)$ by the proof of Lemma~\ref{A3(n,6,4)}.
\end{itemize}

Based on the above analysis, we consider the construction in three cases:  $t= 0,1,2$, corresponding to $n\equiv 0,1,2 \pmod{3}$. Details of recursive constructions and definitions of good short codes are provided separately.

\subsubsection{Case when $t=0$}
\
\newline
\indent In this case, we will use a $4$-GDD of type $g^um^1$ to construct optimal $(n,6,4)$ codes for all $n\equiv 0\pmod{3}$ except for some small $n$.

\begin{theorem}\label{36m}
  Suppose there exists a $4$-GDD of type $g^um^1$, where $g\equiv 0,3,4,7\pmod{12}$. If $A_3(g,6,4)=U(g)$ and $A_3(m,6,4)=U(m)$, then $A_3(gu+m,6,4)=U(gu+m)$.
\end{theorem}

\begin{proof} Let $n=gu+m$.
  Given a $4$-GDD $(X,\G,\B)$ of type $g^um^1$, with $|X|=n$, we construct an $(n,6,4)_3$ code $\C\subset X\times [2]$ as follows. For each group $G\in \G$, construct an optimal $(|G|,6,4)_3$ code ${\C}_G\subset G\times [2]$, which exists by assumption. Note that we can view ${\C}_G$ as a subset of $X\times [2]$, i.e., an $(n,6,4)_3$ code by assigning zeros to the remaining coordinates. Let $\C_0$ be the set of codewords of type $1^4$ obtained from all blocks of $\B$, that is $\C_0=\{\{a_1,b_1,c_1,d_1\}:\{a,b,c,d\}\in \B\}$. Then it is easy to check that $\C=\C_0\bigcup(\cup_{G\in \G}{\C}_G)$ is also an $(n,6,4)_3$ code. Further, since the formula $\frac{g(g+5)}{12}$ is an integer when $g\equiv 0,3,4,7\pmod{12}$, the size of $\C$ is computed as follows.
  \begin{align*}
    |\C|&=|\B|+u\cdot U(g)+U(m)\\
    &=\frac{(g(u-1)+m)gu+gum}{12}+u\cdot\frac{g(g+5)}{12}+\left\lfloor\frac{m(m+5)}{12}\right\rfloor\\
    &=\frac{gu(gu+2m+5)}{12}+\left\lfloor\frac{m(m+5)}{12}\right\rfloor\\
    &=\left\lfloor\frac{(gu+m)(gu+m+5)}{12}\right\rfloor=U(gu+m).
  \end{align*}This completes the proof.
\end{proof}

Note that in Theorem~\ref{36m}, the short code $(g, 6,4)_3$ only requires optimal size. In fact, when $g\equiv 0,3,4,7\pmod{12}$,  a $(g, 6,4)_3$ code of size achieving the upper bound $U(g)$ has no codeword of type $2^2$ by Remark~\ref{case_z}. In particular, it has $g$ codewords of type $1^22^1$ and $U(g)-g$ codewords of  type $1^4$. So its corresponding packing has an empty leave graph by simple computation. The following example gives a specific construction of such a code.

\begin{example}\label{eg36}
   The optimal $(36,6,4)_3$ code of size $U(36)$ is generated from the following base codewords by adding $6 $ modulo $36$ on the set of coordinates. Codewords in short orbits are in bold.

 \[\begin{array}{llll}
     \bm{\{2_1,11_1,20_1,29_1\}}&\bm{\{0_1,9_1,18_1,27_1\}}&\{19_1,22_1,26_1,27_1\}&\{0_1,19_1,23_1,29_1\}\\
     \bm{\{1_1,10_1,19_1,28_1\}}&\{6_1,16_1,18_1,22_1\}&\{9_1,14_1,20_1,31_1\}&\{3_1,10_1,13_1,33_1\}\\
     \{4_1,15_1,17_1,19_1\}&\{2_1,4_1,23_1,35_1\}&\{8_1,15_1,16_1,30_1\}&\{2_1,7_1,22_1,30_1\}\\
     \{2_1,16_1,17_1,28_1\}&\{2_1,5_1,6_1,21_1\}&\{1_1,2_1,12_1,14_1\}&\{6_1,7_1,9_1,23_1\}\\
     \{3_1,35_1,27_2\}&\{21_1,24_1,8_2\}&\{9_1,35_1,28_2\}&\{0_1,31_1,7_2\}\\
     \{1_1,31_1,23_2\}&\{24_1,29_1,18_2\}.\\
    \end{array}
  \]
\end{example}

By Corollary~\ref{coro_4GDD}, a $4$-GDD of type $36^um^1$  exists for $u\geq 4$ and $m\equiv 0\pmod{3}$ with $0\leq m\leq 18(u-1)$. Applying Theorem~\ref{36m} with $g=36$ and the optimal $(36,6,4)_3$ code in Example~\ref{eg36}, we can obtain optimal codes of length $n=36u+m$ if the optimal $(m,6,4)_3$ code exists. Furthermore, when $m$ runs over all representatives of $ 0,3,6,\ldots,33\pmod{36}$, we can obtain all optimal codes of length $n\equiv 0\pmod{3}$ with only a few small cases left.

\subsubsection{Case when $t=1$}
\
\newline
\indent In this case, we deal with $n\equiv 1\pmod{3}$ by using $4$-GDDs with an extra point. By the property of good short codes, we need an optimal $(n,6,4)_3$ code $\C$  of \emph{Property (A)}:  $|\C|=U(n)$ and $\C$ contains exactly $n-1$ codewords of type $1^22^1$ and no codewords of type $2^2$, and the rest of type $1^4$. In fact, a code $\C$ has Property (A) and the leave graph of its corresponding packing is empty  only when $n\equiv 1,6,9,10\pmod{12}$. To prove this, one only needs to check that for what values of $n$, the following equality holds: $(U(n)-(n-1))\times 6+(n-1)\times 3=\binom{n}{2}$, i.e., the number of $2$-subsets contained in all codewords equals $\binom{n}{2}$. Note that the $(10,6,4)_3$ code given in Lemma~\ref{small_base} is a code with Property (A). The following example gives another explicit $(13,6,4)_3$ code with Property (A).


\begin{example}\label{eg13}
  The optimal $(13,6,4)_3$ code with Property (A) listed below  is obtained by computer search.

 \[\begin{array}{llll}
     \{9_1,10_1,11_1,12_1\}&\{6_1,7_1,8_1,12_1\}&\{2_1,3_1,8_1,11_1\}&\{0_1,1_1,2_1,12_1\}\\
     \{3_1,4_1,5_1,12_1\}&\{1_1,5_1,7_1,9_1\}&\{0_1,4_1,6_1,10_1\}&\{5_1,11_1,6_2\}\\
     \{7_1,10_1,3_2\}&\{1_1,4_1,11_2\}&\{2_1,5_1,10_2\}&\{0_1,11_1,7_2\}\\
     \{3_1,6_1,1_2\}&\{1_1,10_1,8_2\}&\{2_1,6_1,9_2\}&\{8_1,9_1,4_2\}\\
     \{0_1,8_1,5_2\}&\{4_1,7_1,2_2\}&\{3_1,9_1,0_2\}.\\
    \end{array}
  \]
\end{example}

The following theorem gives the detail of a construction using $4$-GDDs by adjoining one extra point and inputting short codes with Property (A). Corollary~\ref{coro_4GDD} provides us a good $4$-GDD of type $12^um^1$ for $u\geq 4$ and $m\equiv 0\pmod{3}$ with $0\leq m\leq 6(u-1)$. Since a $(13,6,4)_3$ code with Property (A) exists in Example~\ref{eg13}, pick $g=12$ in Theorem~\ref{12m}, we can obtain all optimal codes of length $n=12u+m+1$ if an optimal $(m+1,6,4)_3$ code exists. Furthermore, when $m+1$ runs over all representatives of $ 1,4,7,10\pmod{12}$, we can obtain all optimal codes of length $n\equiv 1\pmod{3}$ with only a few small cases left.

\begin{theorem}\label{12m}
  Suppose there exists a $4$-GDD of type $g^um^1$ where $g\equiv 0,5,8,9\pmod{12}$. If there exists an optimal $(g+1,6,4)_3$ code with Property (A) and $A_3(m+1,6,4)=U(m+1)$, then $A_3(gu+m+1,6,4)=U(gu+m+1)$.

\end{theorem}

\begin{proof} Suppose that $(X',\G,\B)$ is a $4$-GDD of type $g^um^1$. Let $X=X'\cup \{\infty\}$. We will construct an optimal code $\C$ of length $gu+m+1$ in $X\times [2]$ as follows. For each $G\in \G$  of size $g$, construct an optimal $(g+1,6,4)_3$ code $\C_G\subset (G\cup \{\infty\})\times [2]$ with Property (A), such that the $g$ codewords of type  $1^22^1$ have symbol $2$ in the $g$ positions from $G$, and never in $\infty$. For the group $G$ of size $m$, let $\C_G$ be the optimal  $(m+1,6,4)_3$ code in $(G\cup \{\infty\})\times [2]$.
Next, we view $\C_G$ as a code in $X\times [2]$ in a natural way for each $G\in \G$. Finally, let $\C_0$ be the collection of codewords of type $1^4$ obtained from $\B$. Then $\C=\C_0\bigcup(\cup_{G\in\G}\C_G)$ is a $(gu+m+1,6,4)_3$ code in $X\times [2]$ of size
  \begin{align*}
    |\C|&=|\B|+u\cdot U(g+1)+U(m+1)\\
    &=\frac{(g(u-1)+m)gu+gum}{12}+u\cdot\frac{(g+1)(g+6)-6}{12}+\left\lfloor\frac{(m+1)(m+6)}{12}\right\rfloor\\
    &=\frac{gu(gu+2m+7)}{12}+\left\lfloor\frac{(m+1)(m+6)}{12}\right\rfloor\\
    &=\left\lfloor\frac{(gu+m+1)(gu+m+6)}{12}\right\rfloor=U(gu+m+1).
  \end{align*}
  In the second equality, we use the fact that $U(g+1)=\frac{(g+1)(g+6)-6}{12}$ for each $g\equiv 0,5,8,9\pmod{12}$.
\end{proof}


\subsubsection{Case when $t=2$}
\
\newline
\indent In this case, we construct the optimal $(n,6,4)_3$ codes for $n\equiv 2\pmod{3}$ by using $4$-GDDs with two extra points. Hence, we need a good short $(n,6,4)_3$ code $\C$ of \emph{Property (B)}:  $|\C|=U(n)$ and $\C$ contains exactly  $(n-2)$ codewords of type $1^22^1$, exactly one codeword of type $2^2$,  and the rest of type $1^4$. Similar to the analysis of Property (A), $\C$ has Property (B) and the leave graph of its corresponding packing is empty only when $n\equiv 2,5\pmod{12}$. Delete the codeword of type $2^2$, then the leave graph of its corresponding packing has only one edge. The following example gives an explicit $(26,6,4)_3$ code with Property (B).

\begin{example}\label{eg26}
  The base codewords for an optimal $(26,6,4)_3$ code with Property (B) are listed as follows, which are developed under the automorphism $(0$  $6$  $12$  $18)(1$ $7$  $13$  $19)(2$  $8$ $14$  $20)(3$  $9$  $15$ $21)(4$  $10$  $16$  $22)(5$ $11$  $17$  $23)(24$ $25)$ repeatedly. Codewords in short orbits are in bold. Note that the only codeword of  type $2^2$ is $\{24_2,25_2\}$.

 \[\begin{array}{llll}
     \bm{\{3_1,9_1,15_1,21_1\}}&\bm{\{4_1,10_1,16_1,22_1\}}&\bm{\{5_1,11_1,17_1,23_1\}}&\bm{\{0_1,6_1,12_1,18_1\}}\\
     \bm{\{1_1,7_1,13_1,19_1\}}&\bm{\{2_1,8_1,14_1,20_1\}}&\{4_1,12_1,19_1,25_1\}&\{0_1,10_1,19_1,20_1\}\\
     \{9_1,12_1,16_1,17_1\}&\{8_1,17_1,18_1,25_1\}&\{4_1,6_1,8_1,9_1\}&\{3_1,4_1,17_1,24_1\}\\
     \{7_1,9_1,20_1,24_1\}&\{2_1,5_1,9_1,18_1\}&\{2_1,7_1,10_1,21_1\}&\{1_1,21_1,12_2\}\\
     \{20_1,22_1,17_2\}&\{1_1,5_1,22_2\}&\{5_1,7_1,14_2\}&\{5_1,19_1,3_2\}\\
     \{0_1,17_1,1_2\}&\bm{\{24_2,25_2\}}.\\
    \end{array}
  \]
\end{example}


The good $4$-GDD of type $24^um^1$ exists for $u\geq 4$ and $m\equiv 0\pmod{3}$ with $0\leq m\leq 12(u-1)$ by Corollary~\ref{coro_4GDD}, and a construction of using $4$-GDDs by adding two extra points are given in the following theorem. Together with the good  $(24,6,4)_3$ code in Example~\ref{eg26}, and pick $g=24$ in Theorem~\ref{24m}, we can obtain all optimal codes of length $n=24u+m+2$ if an optimal $(m+2,6,4)_3$ code exists. Furthermore, when $m+2$ runs over  all representatives of $ 2,5,8,\ldots,23\pmod{24}$, we can obtain all optimal codes of length $n\equiv 2\pmod{3}$ with only a few small cases left.

\begin{theorem}\label{24m}
  Suppose there exists a 4-GDD of type $g^um^1$ where $g\equiv 0,3\pmod{12}$. If there exists an optimal $(g+2,6,4)_3$ code  with Property (B) and  $A_3(m+2,6,4)=U(m+2)$, then $A_3(gu+m+2,6,4)=U(gu+m+2)$.
\end{theorem}

\begin{proof} Suppose that $(X',\G,\B)$ is a $4$-GDD of type $g^um^1$, where the specific group $G_0\in \G$ is of size $m$. Let $X=X'\cup \{\imath,\jmath\}$, where $\imath,\jmath\notin X'$. We will construct an optimal code $\C$ of length $gu+m+2$ in $X\times [2]$ as follows. For each $G\in \G$ of size $g$, construct an optimal $(g+2,6,4)_3$ code $\C'_G\subset (G\cup \{\imath,\jmath\})\times [2]$ with Property (B), such that the $g$ codewords of type $1^22^1$ have symbol $2$ in the $g$ positions from $G$, and the codeword of type $2^2$ is $\{\imath_2,\jmath_2\}$. Let $\C_G=\C'_G \setminus\{\{\imath_2,\jmath_2\}\}$ for each $G\in \G$  of size $g$.  For the group $G_0$ of size $m$, let $\C_{G_0}$ be an optimal  $(m+2,6,4)_3$ code in $(G_0\cup \{\imath,\jmath\})\times [2]$.
Next, we view $\C_G$ as a code in $X\times [2]$ in a natural way for each $G\in \G$. Finally, let $\C_0$ be the collection of codewords of type $1^4$ obtained from $\B$. Then $\C=\C_0\cup (\cup_{G\in \G}\C_G)$ is a $(gu+m+2,6,4)_3$ code in $X\times [2]$ of size
%
%
  \begin{align*}
    |\C|&=|\B|+u\cdot(U(g+2)-1)+U(m+2)\\
    &=\frac{(g(u-1)+m)gu+gum}{12}+u\cdot\left(\frac{(g+2)(g+7)-2}{12}-1\right)+\left\lfloor\frac{(m+2)(m+7)}{12}\right\rfloor\\
    &=\frac{gu(gu+2m+9)}{12}+\left\lfloor\frac{(m+2)(m+7)}{12}\right\rfloor\\
    &=\left\lfloor\frac{(gu+m+2)(gu+m+7)}{12}\right\rfloor=U(gu+m+2).
  \end{align*}
   In the second line, we use the fact that $U(g+2)=\frac{(g+2)(g+7)-2}{12}$ for each $g\equiv 0,3\pmod{12}$.
\end{proof}
\subsection{Tables}
In Table~\ref{t_small}, we give the general framework of constructing optimal $(n,6,4)_3$ codes by applying Theorems~\ref{36m}--\ref{24m} with classes of $4$-GDDs and good short codes. Note that the values of $m$ take all representatives modulo the corresponding $g$. For small codes of length $m+t$, all of them can be found in Lemma~\ref{small_base} and Appendix except for $54$ and $60$ in Table~\ref{t_imp}. Since the existence of $4$-GDDs requires $u\geq 4$ except when $m+t=41, 60$ ($u\geq 5$), we need to take care of those codes that can not be covered by Table~\ref{t_small}. We list the lengths of all those codes in the second column of Table~\ref{t_imp}, and deal with them one by one.  The codes whose lengths $n$ are in italic are constructed by applying Theorem~\ref{36m}--\ref{24m} with sporadic $4$-GDDs listed in the third column of Table~\ref{t_imp}. For those lengths in bold, the codes are constructed directly in Appendix. The optimal sizes of the codes with the rest $22$ lengths in normal font are not determined yet, for which the upper and lower bounds are given in Table~\ref{t7}.  The lower bounds in Table~\ref{t7} are deduced in Lemma~\ref{lowerbound} in Appendix~\ref{a4}.


\begin{table}
\center
\caption{The general framework of constructing optimal $(n,6,4)_3$ codes by applying Theorems~\ref{36m}--\ref{24m} with appropriate $4$-GDDs and short codes.  For small codes of length $m+t$, all of them can be found in Lemma~\ref{small_base} and Appendix except for $54$ and $60$ in Table~\ref{t_imp}. Note that when $m+t=41$ or $60$, the corresponding $u$ needs to satisfy that $u\geq 5$ by Corollary~\ref{coro_4GDD}.}\label{t_small}

\begin{tabularx}{14cm}{l|l|X|X|X|X}\hline

\multicolumn{1}{c|}{$n$} &\multicolumn{1}{c|}{$t$} & \multicolumn{1}{c|}{Types of $4$-GDDs}  &  \multicolumn{1}{c|}{ Small codes of length $m+t$}& \multicolumn{1}{c}{Source}\\
\hline
 $n\equiv 0\pmod{3}$ & $t=0$ & \multicolumn{1}{c|}{$36^um^1$, $u\geq 4$ }& $0$, $6$, $9$, $15$, $21$, $27$, $30$, $33$, $39$, $48$, $54$, $60$&\multicolumn{1}{c}{Theorem~\ref{36m}} \\\hline
 $n\equiv 1\pmod{3}$ & $t=1$ & \multicolumn{1}{c|}{$12^um^1$, $u\geq 4$} &$1$, $7$, $10$, $16$& \multicolumn{1}{c}{Theorem~\ref{12m}}\\\hline
 $n\equiv 2\pmod{3}$ & $t=2$ & \multicolumn{1}{c|}{$24^um^1$, $u\geq 4$ }& $8$, $11$, $20$, $23$, $26$, $29$, $38$, $41$& \multicolumn{1}{c}{Theorem~\ref{24m}} \\\hline
\end{tabularx}
\end{table}
%
%
%

\begin{table}
\center
\caption{All remaining codes due to $u\leq 3$ in general  and $u\leq 4$ when $m+t=41, 60$ in Table~\ref{t_small} are listed in the second column.  The codes whose lengths $n$ are in italic are constructed by applying Theorem~\ref{36m}--\ref{24m} with sporadic $4$-GDDs listed in the third column. For those lengths in bold, the codes are constructed directly in Appendix, while the upper and lower bounds for the rest are given in Table~\ref{t7}. }\label{t_imp}
\begin{tabularx}{17cm}{l|X|X|X}\hline
  \multicolumn{1}{c|}{$n$} &\multicolumn{1}{c|}{The remaining small codes} &\multicolumn{1}{c|}{Types of $4$-GDDs} &\multicolumn{1}{c}{Constructions}\\\hline
  $n\equiv 0\pmod{3}$&  $18$,  $24$,   $42$, $\bm{45}$,  $\bm{51}$, $\textit{54}$, $\bm {57}$, $\textit{60}$, $\bm{63}$, $\textit{66}$, $\textit{69}$, $72$, $\textit{75}$, $78$, $\textit{81}$, $84$, $\bm{87}$, $90$, $\bm{93}$, $96$, $\bm{99}$, $102$, $\textit{105}$, $\textit{108}$, $\bm{111}$, $\textit{114}$, $\textit{117}$, $\textit{120}$, $\bm{123}$, $\textit{126}$, $\textit{129}$, $\textit{132}$, $\textit{135}$, $\textit{138}$, $\textit{141}$, $\textit{147}$, $\textit{156}$, $\textit{162}$, $\textit{168}$, $\textit{204}$
	&$15^39^1$, $15^4$, $15^46^1$, $15^49^1$, $15^5$, $15^421^1$, $15^530^1$, $27^4$, $15^79^1$, $27^49^1$, $15^8$, $15^86^1$, $15^89^1$, $15^727^1$, $15^9$, $15^733^1$, $15^821^1$, $15^827^1$, $15^836^1$, $39^46^1$, $15^848^1$, $15^{11}39^1$
	& Theorem~\ref{36m}\\\hline
$n\equiv 1\pmod{3}$& $\bm{19}$, $\bm{22}$, $\bm{25}$, $\textit{28}$, $\bm{31}$, $\bm{34}$, $\textit{37}$, $\bm{40}$, $\textit{43}$, $\textit{46}$, $\textit{52}$&$7^4$, $9^4$, $9^46^1$, $9^5$,$9^56^1$& Theorem~\ref{12m} (one exception $n=28$ using Theorem~\ref{36m})\\\hline
$n\equiv 2\pmod{3}$& $14$, $17$,   $\bm{32}$, $35$,  $44$, $47$, $\bm{50}$, $\bm{53}$, $56$, $59$, $\bm{62}$, $\bm{65}$, $68$, $71$, $\bm{74}$, $\bm{77}$, $80$, $83$, $\bm{86}$, $\bm{89}$, $92$, $95$, $\bm{98}$, $\bm{101}$, $\textit{110}$, $\bm{113}$, $\textit{137}$&$27^4$, $27^5$&Theorem~\ref{24m} (a $(29,6,4)_3$ code with Property (B) is given in Table~\ref{t3_B2})\\\hline
\end{tabularx}
\end{table}

\begin{table*}[h!]
\center
\caption{The lower  and upper bounds of $A_3(n,6,4)$ for small $n$.}\label{t7}
\vspace{-0.3cm}
\[\begin{array}{c|c|c|c|c|c|c|c|c|c|c|c}
\hline
n & 14 & 17 & 18 & 24 & 35 & 42 & 44 & 47 & 56 & 59 & 68 \\
\hline
\text{lower bound} & 21 & 30 & 33 & 55 & 114 & 161 & 176 & 200 & 280 & 310 & 409 \\
\hline
\text{upper bound} & 22 & 31 & 34 & 58 & 116 & 164 & 179 & 203 & 284 & 314 & 413 \\
\hline
\hline
n & 71 & 72 & 78 & 80 & 83 & 84 & 90 & 92 & 95 & 96 & 102\\
\hline
\text{lower bound} & 445 & 461 & 538 & 562 & 603 & 616 & 705 & 738 & 786 & 803 & 901 \\
\hline
\text{upper bound} & 449 & 462 & 539 & 566 & 608 & 623 & 712 & 743 & 791 & 808 & 909 \\
\hline
\end{array}\]
\end{table*}

We summarize our main results of this section as follows.
\begin{theorem}
Let $M=\{14$, $17$, $18$, $24$, $35$, $42$, $44$, $47$, $56$, $59$, $68$, $71$, $72$, $78$, $80$, $83$, $84$, $90$, $92$, $95$, $96$, $102\}$. For any positive integer  $n$,
  \begin{equation*}
    A_3(n,6,4)=\left\{
    \begin{array}{ll}
      U(n)-1,& \text{if } n=3,4,5,12\\
      U(n),& \text{if } n\not\in M \cup\{3,4,5,12\}.
    \end{array}
    \right.
\end{equation*}
For $n\in M$, the lower  and upper bounds for $A_3(n,6,4)$ are given in the Table~\ref{t7}.

\end{theorem}

%
%

\section{Ternary CWCs of Distance $2w-2$}\label{dis2w}

In this section, 
we consider  ternary CWCs with weight $w$ and distance $2w-2$ for general $w$,  and determine the value of $A_3(n,2w-2,w)$ when $n$ is sufficiently large under certain conditions based on graph packings. Note that we do not give a lower bound on the length $n$ for which $A_3(n,2w-2,w)$ is determined, compared with the explicit constructions for $w=3,4$. One interesting thing is that, although there are  vectors of growing types when $w$ increases, the codes with maximum size we construct will have only two different types of codewords: $1^w$ and $1^{w-2}2^1$.


  For a graph $H$ without isolated vertices,  $\gcd(H)$ denotes the greatest common divisor of the degrees of all vertices of $H$. A graph $G$ is called {\it $d$-divisible} if  $ \gcd(G)$ is divisible by $d$, while $G$ is called {\it nowhere $d$-divisible} if no vertex of $G$ has degree divisible by  $d$. An {\it $H$-packing} of a graph $G$ is a set $\{G_1,\ldots,G_s\}$ of edge-disjoint subgraphs of $G$ where each subgraph is isomorphic to $H$. Further, if $G$ is a union of $G_i$, $i=1,2,\ldots,s$, then we call it an {\it $H$-decomposition}. The {\it $H$-packing number} of $G$, denoted by $P(H,G)$, is the maximum cardinality of an $H$-packing of $G$.   In particular, if $H=K_k$ and $G=K_n$ are  complete graphs, then $P(H,G)$ is equal to the packing number $D(n,k,2)$ introduced in Section~\ref{sec:iib}.
%
Our main tool is the following result of Alon et al. \cite{alon1998packing}.

\begin{theorem}\label{graph}
  Let $H$ be a graph with $h$ edges, and let $\gcd(H)=e$. Then there exist $N=N(H)$, and $\varepsilon=\varepsilon(H)$ such that for any $e$-divisible or nowhere $e$-divisible graph $G=(V,E)$ with $n>N(H)$ vertices and $\delta(G)>(1-\varepsilon(H))n$, \[P(H,G)=\left\lfloor\frac{\sum_{v\in V}\alpha_v}{2h}\right\rfloor,\] unless when $G$ is $e$-divisible and $0<|E|\pmod{h}\leq \frac{e^2}{2}$, in which case  \[P(H,G)=\left\lfloor\frac{\sum_{v\in V}\alpha_v}{2h}\right\rfloor-1.\]  Here, $\alpha_v$ is the degree of vertex $v$, rounded down to the closest multiple of $e$.
\end{theorem}


Consider a ternary code $\C\subset I_3^{n}$  with constant weight $w$. A codeword in $\C$ has \emph{type}  $1^x2^y$ with $x+2y=w$ for some  $y\in [0,\left\lfloor\frac{w}{2}\right\rfloor]$. By the \textbf{UNC} condition and the  distance formula in Eq. (\ref{eqdis}), $\C$ is an   $(n,2w-2,w)_3$ code if and only if $\C$ satisfies the following properties:
\begin{itemize}
  \item[$(1'')$] The collection of subsets $supp(\textbf{u})\subset\bbZ_n$, for all codewords $\textbf{u}\in\C$, forms a $2$-$(n,\{\left\lceil\frac{w}{2}\right\rceil, \left\lceil\frac{w}{2}\right\rceil+1, \ldots, w\},1)$ packing.
  \item[$(2'')$] For any two codewords $\textbf{u},\textbf{v}\in \C$, if $i\in supp(\textbf{u})\cap supp(\textbf{v})$, then $\min\{\textbf{u}_i,\textbf{v}_i\}=1$.
\end{itemize}


\begin{lemma}\label{general}
  $A_3(n,2w-2,w)\leq\left\lfloor\frac{n\left(n-1-(w-1)(w-2)\right)}{w(w-1)}\right\rfloor+n$.
\end{lemma}

\begin{proof}
Let $\beta_y$ be the number of codewords of type $1^x2^y$,  $y=0,1,\ldots, \left\lfloor\frac{w}{2}\right\rfloor$.
  By properties $(1'')$ and $(2'')$, we have
  \begin{align}
   \binom{w}{2}\beta_0+\binom{w-1}{2}\beta_1+\cdots+\binom{w-\left\lfloor\frac{w}{2}\right\rfloor}{2}\beta_{\left\lfloor\frac{w}{2}\right\rfloor} &\leq \binom{n}{2}, \text{ and}\label{(13)}\\
    \beta_1+2\beta_2+\cdots+\left\lfloor\frac{w}{2}\right\rfloor \beta_{\left\lfloor\frac{w}{2}\right\rfloor}& \leq n.\label{(14)}
  \end{align}
 Note that $\binom{w-t}{2}+t(w-1)=\binom{w}{2}+\binom{t}{2}$.  Computing  $(\ref{(13)})+(w-1)(\ref{(14)})$, we obtain
  \begin{equation*}
    \beta_0+\beta_1+\cdots+\beta_{\left\lfloor\frac{w}{2}\right\rfloor}\leq\left\lfloor\frac{n\left(n-1-(w-1)(w-2)\right)}{w(w-1)}\right\rfloor+n,
  \end{equation*}which completes the proof.
\end{proof}

The upper bound given in Lemma~\ref{general} is consistent with the cases for  $w=3$ and $4$. For convenience, let \[B(n):=\left\lfloor\frac{n\left(n-1-(w-1)(w-2)\right)}{w(w-1)}\right\rfloor.\] We will prove that the upper bound in Lemma~\ref{general} can be achieved for certain values of $n$ by Theorem~\ref{graph}. The desired code $\C$ has only two types of codewords, $1^w$ and $1^{w-2}2^1$.  Consider the complete graph $K_n$ with vertex set $\bbZ_n$. For each $\textbf{u}\in \C$, view it as a complete subgraph on the vertex set $supp(\textbf{u})$, this is an injective mapping since we do not know the symbol is $1$ or $2$ in the position from $supp(\textbf{u})$. Then property $(1'')$ tells that all these subgraphs are pairwise edge-disjoint, thus form a packing of $K_n$. We need to be careful about codewords containing symbol $2$, for which the corresponding positions should be different by property $(2'')$. We construct such codewords by Golomb rulers  \cite{Shearer2006difference}.

An $(n,w)$ \emph{modular Golomb ruler} is a set of $w$ integers $\{a_1 , a_2 ,\ldots, a_w\}$, such that all of
the differences, $\{a_i- a_j\mid 1 \leq i \neq j \leq w\}$, are distinct and nonzero modulo $n$. Suppose that we have an  $(n,w-1)$  modular Golomb ruler $\{a_1, a_2,\ldots, a_{w-1}\}$.  Then the $n$ codewords $\{(a_1+i)_2, (a_2+i)_1,\ldots, (a_{w-1}+i)_1\}$ of type $1^{w-2}2^1$, $i\in \bbZ_n$ have pairwise distance at least $2w-2$. Associate these $n$ codewords with $n$ complete graphs  $K_{w-1}$ with vertex set $\{a_1+i, a_2+i,\ldots, a_{w-1}+i\}$, $i\in \bbZ_n$,  which are edge-disjoint  subgraphs of $K_n$  with vertex set $\bbZ_n$.   Let $S$ be the union of these $n$  subgraphs $K_{w-1}$. It is easy to show that $S$ is a {\it regular} subgraph of $K_n$ with each vertex of same degree $(w-1)(w-2)$. Denote $G=K_n\setminus S$. We will apply Theorem~\ref{graph} to obtain a $K_w$-packing of $G$, which yields the remaining  codewords of type $1^w$.

%

\begin{theorem}\label{thmlarge}
Let $w\geq 3$ be any fixed integer.  Then $A_3(n,2w-2,w)\geq B(n)+n-1$ for any
 sufficiently large integer $n\equiv 1\pmod{w-1}$. Further if $n\equiv 1,w,-w+2,-2w+3\pmod{w(w-1)}$,  then $A_3(n,2w-2,w)=B(n)+n$.
\end{theorem}

\begin{proof} By \cite{Shearer2006difference}, there exists an $(n,w-1)$ modular Golomb ruler  for any $n=\Omega(w^2)$. Thus, by above discussion, we obtain $n$ codewords of type $1^{w-2}2^1$, and a regular graph $G=(V,E)$ with vertex set $V=\bbZ_n$  and degree $d=n-1-(w-1)(w-2)$.  Let $H=K_w$ in Theorem~\ref{graph}, then $e=\gcd(H)=w-1$ and $h=\frac{w(w-1)}{2}$. Since $n$ is  sufficiently large, we have
$d>(1-\varepsilon)n$, where $\varepsilon=\varepsilon(K_w)$  is defined in Theorem~\ref{graph}. Further $n\equiv 1\pmod{w-1}$ implies that $d\equiv 0\pmod{w-1}$, i.e., $G$ is $(w-1)$-divisible, so $\alpha_v=d=n-1-(w-1)(w-2)$ for each $v\in V$. By Theorem~\ref{graph}, we have a $K_w$-packing of $G$ with packing number

  \begin{align*}
    P(K_w,G)&\geq \left\lfloor\frac{\sum_{v\in V}\alpha_v}{2\binom{w}{2}}\right\rfloor-1=\left\lfloor\frac{n\left(n-1-(w-1)(w-2)\right)}{w(w-1)}\right\rfloor-1=B(n)-1.
  \end{align*}
 Each $K_w$ of this packing gives a codeword of type $1^w$ in a natural way. Thus we obtain at least $B(n)-1$ codewords. Combining the $n$ codewords of type $1^{w-2}2^1$, we have an $(n,2w-2,w)_3$ code of size at least $B(n)+n-1$.

When $n\equiv w,-2w+3\pmod{w(w-1)}$,  it is easy to check that  $|E|=nd/2\equiv 0\pmod{\frac{w(w-1)}{2}}$. In fact, one can show that these are the only two congruent classes of $n$ that satisfy $|E| \pmod{\frac{w(w-1)}{2}}\notin [1,\frac{(w-1)^2}{2}]$. By Theorem~\ref{graph},  we have $P(K_w,G)=B(n)$, hence $A_3(n,2w-2,w)=B(n)+n$.

  For $n\equiv 1,-w+2\pmod{w(w-1)}$, we consider a slightly different model. Let $S'$ be a regular graph on $\bbZ_{n-1}$ of degree $(w-1)(w-2)$, which is a union of $n-1$ edge-disjoint complete subgraphs $K_{w-1}$ obtained by an $(n-1,w-1)$ modular Golomb ruler. Note that $S'$ corresponds to $n-1$ codewords of  type $1^{w-2}2^1$. Let $G'=K_n\setminus S'$, which has vertex set
  $V'=\bbZ_{n-1}\cup\{\infty\}$ and edge set $E'$.
 Then $G'$ has degree $n-1-(w-1)(w-2)$ for each $v\in\bbZ_{n-1}$,  and degree $n-1$ for the vertex  $\infty$. So $G'$ is $(w-1)$-divisible.  Further, it is easy to check that \[|E'|=\frac{(n-1)\left(n-1-(w-1)(w-2)\right)+n-1}{2}\equiv 0\pmod{\frac{w(w-1)}{2}}.\] By Theorem~\ref{graph}.  we have
  \begin{align*}
    P(K_w,G^\prime)&=\left\lfloor\frac{\sum_{v\in V}\alpha_v}{2\binom{w}{2}}\right\rfloor=\frac{n\left(n-1-(w-1)(w-2)\right)+(w-1)(w-2)}{w(w-1)}=B(n)+1.
  \end{align*}
  The last equality is due to the fact that $B(n)=\frac{n\left(n-1-(w-1)(w-2)\right)-2(w-1)}{w(w-1)}$ in these cases.
 Hence we get $A_3(n,2w-2,w)=n-1+ P(K_w,G^\prime)=B(n)+n$ for $n\equiv 1,-w+2\pmod{w(w-1)}$.
\end{proof}

Theorem~\ref{thmlarge} tells us when $n$ is sufficiently large and $n\equiv 1,w,-w+2,-2w+3\pmod{w(w-1)}$, $A_3(n,2w-2,w)=B(n)+n$. In particular, $A_3(n,6,4)=B(n)+n=U(n)$ with $n\equiv 1\pmod{3}$ for $w=4$, and $A_3(n,4,3)=B(n)+n=\left\lfloor\frac{n^2+3n}{6}\right\rfloor$ with $n\equiv 1\pmod{2}$ for $w=3$. These results are consistent with the previous sections when $w\leq 4$ and $d=2w-2$, and are weaker than them since we have proved that $A_3(n,2w-2,w)=B(n)+n$ for all $n$ and $w\leq 4$ before, with very few exceptions for $w=4$.


\section{Conclusion}\label{s:con}
In this paper, we determine the maximum size of constant-weight codes in $\ell_1$-metric over non-negative integers or over $I_3=\{0,1,2\}$, for weight three and four. We also provide an asymptotic result for the maximum size of ternary codes with general weight $w$ and distance $2w-2$. It is plausible that we could extend the method in Section~\ref{dis2w} by looking for irregular graphs $S$, such that by deleting edges in $S$, the resultant has a clique decomposition. We leave this for future study. Further, constructing optimal constant-weight  codes over non-negative integers is a more challenging problem.


\section*{Acknowledgement}


The authors express their gratitude to the anonymous reviewers  and the associate editor for their detailed and constructive comments which
are very helpful to the improvement of the technical presentation of this paper.

\section*{Appendix}
Here, we present all small optimal codes  in Tables~\ref{t_small} and \ref{t_imp} that are constructed by computer search. Codewords in short orbits are in bold for all tables in this section. For easy verification of the codewords listed here, we upload some code files on GitHub so that interested readers can check the correctness, see \href{https://github.com/llxu2020/verification}{https://github.com/llxu2020/verification}.

\subsection{The small optimal codes of length $n\equiv 0\pmod{3}$}\label{a1}

For each $n\in\{15$, $21$, $27$, $30$, $33$, $39$, $45$, $48$, $51$, $57$, $63$, $87$, $93$, $99$, $111$, $123\}$, the base codewords of the optimal codes in $\bbZ_n\times [2]$ are listed in Table~\ref{t3_B0} but with different group actions.  For $n= 21$, $33$ and $45$, the automorphism is $(0$ $4$ $8$ $\cdots$ $n-5)(1$ $5$ $9$ $\cdots$ $n-4)(2$ $6$ $10$ $\cdots$ $n-3)(3$ $7$ $11$ $\cdots$ $n-2)(n-1)$; for $n=30$, the automorphism is $(0$ $4$ $8$ $\cdots$ $20)(1$ $5$ $9$ $\cdots$ $21)(2$ $6$ $10$ $\cdots$ $22)(3$ $7$ $11$ $\cdots$ $23)(24$  $25$  $26$  $\cdots$  $29)$; for $n=57$, $93$, the automorphism is $(0$  $2$  $4$  $\cdots$  $n-3)(1$  $3$  $5$  $\cdots$  $n-2)(n-1)$; and for the rest, the codes will be generated by adding $3\pmod{n}$ on the set of coordinates except for $48$, in which case it will be developed by  adding $6 $ modulo $48$.
\\

\setlength{\tabcolsep}{2pt}
\begin{footnotesize}
\begin{longtable}{l|llllll}
\caption{Base codewords of small CWCs of length $n\equiv 0\pmod{3}$.}\label{t3_B0}\\
\hline
\multicolumn{1}{c|}{$n$} & \multicolumn{6}{c}{\text{Base Codewords}}\\\hline
\endfirsthead
\caption{Base codewords of small CWCs of length $n\equiv 0\pmod{3}$. Codewords in short orbits are in bold, as are all other codes in this paper found by computer.}\\
\hline
\multicolumn{1}{c|}{$n$} & \multicolumn{6}{c}{\text{Base Codewords}}\\\hline
\endhead
\hline\\ \endfoot
\hline\endlastfoot
 \multirow{1}*{15} &$\{0_1,5_1,8_1,9_1\}$&$\{1_1,4_1,8_1,14_1\}$&$\{0_1,1_1,2_2\}$&$\{0_1,3_1,10_2\}$&$\{1_1,7_1,3_2\}$ \\
\hline \multirow{2}*{21} &$\{7_1,9_1,14_1,19_1\}$&$\{3_1,4_1,6_1,7_1\}$&$\{5_1,14_1,17_1,18_1\}$&$\{0_1,7_1,12_1,13_1\}$&$\{0_1,5_1,9_1,11_1\}$&$\{8_1,18_1,4_2\}$\\
 &$\{4_1,20_1,1_2\}$&$\{10_1,20_1,19_2\}$&$\{8_1,14_1,6_2\}$\\
\hline \multirow{2}*{27}&$\{3_1,7_1,19_1,21_1\}$&$\{2_1,14_1,18_1,21_1\}$&$\{5_1,8_1,16_1,22_1\}$&$\{1_1,2_1,23_1,25_1\}$&$\{2_1,3_1,15_1,20_1\}$&$\{12_1,19_1,18_2\}$\\
 & $\{7_1,15_1,25_2\}$&$\{1_1,6_1,8_2\}$\\
\hline \multirow{3}*{30} &$\{16_1,23_1,28_1,29_1\}$&\bm{$\{0_1,6_1,12_1,18_1\}$}&$\{17_1,18_1,27_1,29_1\}$&$\{1_1,11_1,12_1,21_1\}$&$\{6_1,10_1,15_1,17_1\}$&$\{4_1,12_1,14_1,29_1\}$\\
    &\bm{$\{1_1,7_1,13_1,19_1\}$}&$\{1_1,3_1,10_1,26_1\}$&$\{5_1,8_1,13_1,25_1\}$&$\{1_1,6_1,8_1,22_1\}$&$\{2_1,3_1,23_1,27_1\}$&$\{16_1,25_1,17_2\}$\\
    &$\{12_1,23_1,15_2\}$&$\{15_1,20_1,0_2\}$&$\{3_1,25_1,14_2\}$&\bm{$\{24_2,27_2\}$}\\
\hline \multirow{3}*{33} &\bm{$\{2_1,10_1,18_1,26_1\}$}&\bm{$\{7_1,15_1,23_1,31_1\}$}&$\{9_1,27_1,30_1,31_1\}$&$\{0_1,12_1,26_1,31_1\}$&$\{16_1,17_1,18_1,27_1\}$&\bm{$\{1_1,9_1,17_1,25_1\}$}\\
    &$\{0_1,10_1,23_1,30_1\}$&$\{1_1,4_1,10_1,31_1\}$&\bm{$\{0_1,8_1,16_1,24_1\}$}&$\{5_1,10_1,20_1,24_1\}$&$\{0_1,3_1,9_1,15_1\}$&$\{0_1,5_1,7_1,25_1\}$\\
    &$\{14_1,32_1,31_2\}$&$\{5_1,18_1,1_2\}$&$\{18_1,25_1,22_2\}$&$\{21_1,32_1,0_2\}$\\
\hline \multirow{2}*{39}&$\{23_1,27_1,30_1,34_1\}$&$\{15_1,28_1,30_1,36_1\}$&$\{13_1,14_1,23_1,31_1\}$&$\{0_1,22_1,26_1,38_1\}$&$\{10_1,24_1,33_1,35_1\}$&$\{8_1,11_1,13_1,33_1\}$\\
 &$\{0_1,5_1,23_1,29_1\}$&$\{18_1,19_1,30_1,38_1\}$&$\{10_1,19_1,7_2\}$&$\{19_1,34_1,24_2\}$&$\{7_1,13_1,20_2\}$\\
\hline
\multirow{3}*{45} &$\{16_1,21_1,22_1,33_1\}$&$\{8_1,19_1,37_1,43_1\}$&$\{10_1,34_1,35_1,37_1\}$&$\{15_1,23_1,28_1,37_1\}$&$\{9_1,20_1,22_1,34_1\}$&$\{6_1,34_1,39_1,40_1\}$\\
 &$\{5_1,15_1,25_1,41_1\}$&$\{7_1,22_1,30_1,37_1\}$&$\{2_1,4_1,19_1,25_1\}$&$\{0_1,16_1,20_1,23_1\}$&$\{0_1,25_1,27_1,34_1\}$&$\{0_1,8_1,26_1,30_1\}$\\
  &$\{3_1,31_1,34_1,43_1\}$&$\{17_1,44_1,22_2\}$&$\{16_1,44_1,35_2\}$&$\{25_1,28_1,21_2\}$&$\{36_1,37_1,24_2\}$\\
\hline
\multirow{6}*{48}&\bm{$\{5_1,17_1,29_1,41_1\}$}&\bm{$\{0_1,12_1,24_1,36_1\}$}&\bm{$\{1_1,13_1,25_1,37_1\}$}&\bm{$\{2_1,14_1,26_1,38_1\}$}&\bm{$\{3_1,15_1,27_1,39_1\}$}&\bm{$\{4_1,16_1,28_1,40_1\}$}\\
 &$\{13_1,17_1,23_1,42_1\}$&$\{10_1,14_1,19_1,28_1\}$&$\{26_1,28_1,30_1,35_1\}$&$\{10_1,20_1,27_1,47_1\}$&$\{19_1,27_1,35_1,45_1\}$&$\{0_1,15_1,22_1,28_1\}$\\
  &$\{19_1,25_1,42_1,44_1\}$&$\{13_1,15_1,44_1,47_1\}$&$\{4_1,15_1,25_1,26_1\}$&$\{10_1,11_1,29_1,42_1\}$&$\{11_1,12_1,15_1,38_1\}$&$\{1_1,14_1,29_1,34_1\}$\\
   &$\{8_1,16_1,19_1,24_1\}$&$\{0_1,6_1,13_1,20_1\}$&$\{0_1,27_1,30_1,41_1\}$&$\{6_1,15_1,16_1,43_1\}$&$\{1_1,23_1,31_1,45_1\}$&$\{6_1,7_1,10_1,39_1\}$\\
    &$\{8_1,18_1,26_1,39_1\}$&$\{9_1,14_1,15_2\}$&$\{10_1,41_1,43_2\}$&$\{15_1,17_1,40_2\}$&$\{34_1,39_1,0_2\}$&$\{11_1,26_1,20_2\}$\\
     &$\{10_1,26_1,23_2\}$\\
\hline
\multirow{3}*{51}&$\{23_1,27_1,48_1,50_1\}$&$\{32_1,34_1,44_1,47_1\}$&$\{6_1,39_1,44_1,50_1\}$&$\{26_1,27_1,36_1,42_1\}$&$\{9_1,26_1,31_1,48_1\}$&$\{4_1,7_1,25_1,27_1\}$\\
 &$\{2_1,25_1,34_1,39_1\}$&$\{7_1,15_1,39_1,47_1\}$&$\{0_1,1_1,7_1,48_1\}$&$\{0_1,16_1,20_1,40_1\}$&$\{8_1,22_1,29_1,37_1\}$&$\{20_1,46_1,11_2\}$\\
  &$\{5_1,23_1,22_2\}$&$\{4_1,16_1,42_2\}$\\
\hline \multirow{2}*{57} &$\{39_1,40_1,46_1,48_1\}$&$\{0_1,22_1,23_1,53_1\}$&\bm{$\{0_1,14_1,28_1,42_1\}$}&\bm{$\{1_1,15_1,29_1,43_1\}$}&$\{20_1,24_1,39_1,49_1\}$&$\{5_1,10_1,13_1,28_1\}$\\
 &$\{0_1,17_1,21_1,37_1\}$&$\{9_1,22_1,31_1,33_1\}$&$\{10_1,20_1,50_1,55_1\}$&$\{4_1,16_1,36_1,43_1\}$&$\{7_1,25_1,13_2\}$&$\{17_1,56_1,4_2\}$\\
\hline \multirow{3}*{63} &$\{14_1,22_1,48_1,59_1\}$&$\{8_1,19_1,27_1,59_1\}$&$\{1_1,13_1,44_1,60_1\}$&$\{26_1,47_1,61_1,62_1\}$&$\{18_1,37_1,42_1,44_1\}$&$\{11_1,12_1,18_1,33_1\}$\\
 &$\{18_1,45_1,52_1,58_1\}$&$\{4_1,8_1,46_1,62_1\}$&$\{8_1,37_1,47_1,57_1\}$&$\{1_1,16_1,54_1,55_1\}$&$\{21_1,26_1,43_1,56_1\}$&$\{12_1,24_1,32_1,57_1\}$\\
  &$\{0_1,9_1,59_1,61_1\}$&$\{6_1,9_1,37_1,55_1\}$&$\{10_1,29_1,32_2\}$&$\{52_1,55_1,19_2\}$&$\{26_1,32_1,9_2\}$\\
\hline \multirow{4}*{87}&$\{50_1,69_1,76_1,82_1\}$&$\{31_1,49_1,73_1,74_1\}$&$\{41_1,43_1,74_1,83_1\}$&$\{15_1,38_1,44_1,65_1\}$&$\{10_1,15_1,56_1,67_1\}$&$\{29_1,30_1,46_1,69_1\}$\\
 &$\{19_1,55_1,58_1,66_1\}$&$\{17_1,35_1,69_1,74_1\}$&$\{7_1,29_1,36_1,79_1\}$&$\{6_1,10_1,14_1,38_1\}$&$\{44_1,47_1,52_1,59_1\}$&$\{23_1,43_1,59_1,84_1\}$\\
  &$\{0_1,38_1,51_1,73_1\}$&$\{19_1,40_1,53_1,57_1\}$&$\{2_1,12_1,31_1,40_1\}$&$\{9_1,15_1,36_1,80_1\}$&$\{3_1,4_1,14_1,81_1\}$&$\{6_1,8_1,31_1,81_1\}$\\
   &$\{0_1,15_1,33_1,57_1\}$&$\{4_1,58_1,60_1,77_1\}$&$\{40_1,66_1,13_2\}$&$\{27_1,30_1,86_2\}$&$\{40_1,52_1,72_2\}$\\
\hline \multirow{3}*{93} &\bm{$\{0_1,23_1,46_1,69_1\}$}&$\{2_1,38_1,51_1,78_1\}$&$\{21_1,33_1,49_1,58_1\}$&$\{20_1,25_1,28_1,35_1\}$&$\{8_1,12_1,14_1,26_1\}$&$\{41_1,54_1,55_1,88_1\}$\\
 &$\{1_1,49_1,64_1,85_1\}$&$\{7_1,13_1,67_1,74_1\}$&$\{27_1,62_1,79_1,84_1\}$&$\{7_1,29_1,58_1,88_1\}$&$\{8_1,40_1,59_1,90_1\}$&$\{20_1,48_1,57_1,74_1\}$\\
  &$\{4_1,7_1,31_1,57_1\}$&$\{0_1,20_1,44_1,91_1\}$&$\{3_1,7_1,37_1,56_1\}$&$\{13_1,85_1,11_2\}$&$\{19_1,92_1,30_2\}$\\
 \hline\multirow{5}*{99}&$\{36_1,60_1,68_1,98_1\}$&$\{31_1,41_1,77_1,89_1\}$&$\{14_1,15_1,25_1,46_1\}$&$\{14_1,28_1,73_1,92_1\}$&$\{56_1,66_1,79_1,96_1\}$&$\{14_1,47_1,69_1,90_1\}$\\
 &$\{11_1,30_1,39_1,45_1\}$&$\{19_1,20_1,23_1,94_1\}$&$\{27_1,79_1,81_1,82_1\}$&$\{9_1,38_1,42_1,90_1\}$&$\{12_1,23_1,54_1,79_1\}$&$\{5_1,29_1,73_1,78_1\}$\\
 &$\{49_1,62_1,79_1,88_1\}$&$\{40_1,47_1,56_1,74_1\}$&$\{10_1,30_1,52_1,93_1\}$&$\{3_1,38_1,63_1,67_1\}$&$\{20_1,25_1,36_1,77_1\}$&$\{0_1,49_1,72_1,86_1\}$\\
 &$\{72_1,75_1,77_1,92_1\}$&$\{19_1,33_1,52_1,70_1\}$&$\{3_1,31_1,37_1,90_1\}$&$\{0_1,53_1,61_1,73_1\}$&$\{10_1,32_1,39_1,82_1\}$&$\{68_1,74_1,29_2\}$\\
 &$\{58_1,73_1,66_2\}$&$\{5_1,7_1,43_2\}$\\
 \hline\multirow{5}*{111}&$\{12_1,55_1,57_1,93_1\}$&$\{46_1,67_1,97_1,104_1\}$&$\{55_1,70_1,86_1,103_1\}$&$\{2_1,16_1,91_1,108_1\}$&$\{5_1,78_1,99_1,107_1\}$&$\{33_1,72_1,75_1,100_1\}$\\
&$\{17_1,20_1,63_1,79_1\}$&$\{53_1,71_1,78_1,102_1\}$&$\{11_1,48_1,95_1,99_1\}$&$\{42_1,79_1,90_1,99_1\}$&$\{42_1,76_1,82_1,95_1\}$&$\{34_1,62_1,68_1,101_1\}$\\
&$\{23_1,59_1,61_1,71_1\}$&$\{14_1,33_1,48_1,55_1\}$&$\{2_1,13_1,18_1,31_1\}$&$\{1_1,27_1,85_1,88_1\}$&$\{29_1,34_1,42_1,69_1\}$&$\{0_1,11_1,32_1,82_1\}$\\
&$\{9_1,53_1,98_1,108_1\}$&$\{11_1,68_1,96_1,98_1\}$&$\{4_1,16_1,58_1,59_1\}$&$\{6_1,7_1,16_1,39_1\}$&$\{0_1,35_1,50_1,70_1\}$&$\{0_1,6_1,20_1,52_1\}$\\
&$\{9_1,13_1,27_1,58_1\}$&$\{5_1,28_1,74_1,75_1\}$&$\{10_1,95_1,35_2\}$&$\{41_1,88_1,49_2\}$&$\{97_1,101_1,42_2\}$.\\
\hline\multirow{6}*{123}&$\{32_1,34_1,115_1,118_1\}$&$\{40_1,60_1,85_1,111_1\}$&$\{12_1,40_1,61_1,117_1\}$&$\{0_1,68_1,119_1,122_1\}$&$\{56_1,81_1,116_1,121_1\}$&$\{25_1,55_1,104_1,119_1\}$\\
 &$\{52_1,59_1,111_1,115_1\}$&$\{17_1,26_1,93_1,107_1\}$&$\{30_1,32_1,80_1,111_1\}$&$\{23_1,30_1,87_1,119_1\}$&$\{22_1,97_1,102_1,114_1\}$&$\{9_1,43_1,105_1,116_1\}$\\
 &$\{37_1,56_1,84_1,103_1\}$&$\{13_1,22_1,46_1,56_1\}$&$\{10_1,16_1,28_1,41_1\}$&$\{50_1,80_1,86_1,99_1\}$&$\{13_1,17_1,54_1,119_1\}$&$\{70_1,72_1,78_1,85_1\}$\\
 &$\{36_1,56_1,74_1,109_1\}$&$\{3_1,20_1,32_1,58_1\}$&$\{31_1,59_1,82_1,83_1\}$&$\{6_1,42_1,90_1,100_1\}$&$\{20_1,28_1,66_1,82_1\}$&$\{17_1,57_1,62_1,94_1\}$\\
 &$\{29_1,70_1,86_1,97_1\}$&$\{4_1,15_1,36_1,39_1\}$&$\{3_1,33_1,56_1,66_1\}$&$\{6_1,14_1,28_1,115_1\}$&$\{9_1,10_1,54_1,63_1\}$&$\{21_1,47_1,6_2\}$\\
 &$\{41_1,84_1,61_2\}$&$\{21_1,83_1,122_2\}$\\
\end{longtable}
\end{footnotesize}


\subsection{The small optimal codes of length $n\equiv 1\pmod{3}$}\label{a2}

For each $n\in\{16$, $19$, $22$, $25$, $31$, $34$, $40 \}$,
the base codewords of the optimal codes in $\bbZ_n\times [2]$ are listed in Table~\ref{B_t2} but with different group actions. For $n=16$, $40$, the automorphism is $(0$  $4$  $8$  $\cdots$  $n-4)(1$  $5$  $9$  $\cdots$  $n-3)(2$  $6$  $10$  $\cdots$  $n-2)(3$  $7$  $11$  $\cdots$  $n-1)$; for $n=19$, $31$, the automorphism is $(0$  $1$  $2$  $\cdots$  $n-1)$; for $n=22$, $34$, the automorphism is $(0$  $3$  $6$  $\cdots$  $n-4)(1$  $4$  $7$  $\cdots$  $n-3)(2$  $5$  $8$  $\cdots$  $n-2)(n-1)$; and for $n=25$,  the automorphism is $(0$  $6$  $12$  $18)(1$  $7$  $13$  $19)(2$  $8$  $14$  $20)(3$  $9$  $15$  $21)(4$  $10$  $16$  $22)(5$  $11$  $17$  $23)(24)$.

\setlength{\tabcolsep}{2pt}
\begin{footnotesize}
\begin{longtable}{l|llllll}
\caption{Base codewords of small CWCs of length $n\equiv 1\pmod{3}$.}\label{B_t2}\\
\hline
\multicolumn{1}{c|}{$n$} & \multicolumn{6}{c}{\text{Base Codewords}}\\\hline
\endfirsthead
\caption{Base codewords of small CWCs of length $n\equiv 1\pmod{3}$.}\\
\hline
\multicolumn{1}{c|}{$n$} & \multicolumn{6}{c}{\text{Base Codewords}}\\\hline
\endhead
\hline\\ \endfoot
\hline\endlastfoot
\multirow{2}*{16}&\bm{$\{2_1,6_1,10_1,14_1\}$}&\bm{$\{3_1,7_1,11_1,15_1\}$}&\bm{$\{0_1,4_1,8_1,12_1\}$}&\bm{$\{1_1,5_1,9_1,13_1\}$}&$\{0_1,7_1,10_1,13_1\}$&$\{0_1,9_1,14_1,15_1\}$\\
&$\{2_1,7_1,12_2\}$&$\{1_1,10_1,3_2\}$&$\{0_1,3_1,5_2\}$&$\{0_1,1_1,2_2\}$\\
\hline
 \multirow{1}*{19} &$\{0_1,2_1,5_1,15_1\}$&$\{0_1,1_1,8_2\}$\\
  \hline
  \multirow{2}*{22} &$\{1_1,3_1,9_1,10_1\}$&$\{3_1,5_1,6_1,14_1\}$&$\{1_1,11_1,18_1,21_1\}$&$\{1_1,2_1,4_1,17_1\}$&$\{0_1,16_1,10_2\}$&$\{1_1,5_1,8_2\}$\\
   &$\{3_1,8_1,12_2\}$\\
\hline \multirow{4}*{25} &\bm{$\{4_1,10_1,16_1,22_1\}$}&\bm{$\{5_1,11_1,17_1,23_1\}$}&\bm{$\{0_1,6_1,12_1,18_1\}$}&\bm{$\{1_1,7_1,13_1,19_1\}$}&\bm{$\{2_1,8_1,14_1,20_1\}$}&\bm{$\{3_1,9_1,15_1,21_1\}$}\\
 &$\{1_1,8_1,10_1,23_1\}$&$\{0_1,8_1,15_1,16_1\}$&$\{7_1,10_1,15_1,17_1\}$&$\{3_1,13_1,17_1,18_1\}$&$\{5_1,9_1,20_1,24_1\}$&$\{0_1,4_1,7_1,24_1\}$\\
  &$\{8_1,9_1,12_1,22_1\}$&$\{6_1,8_1,11_1,19_1\}$&$\{0_1,17_1,22_2\}$&$\{2_1,23_1,12_2\}$&$\{3_1,8_1,7_2\}$&$\{4_1,21_1,5_2\}$\\
   &$\{12_1,13_1,15_2\}$&$\{10_1,19_1,14_2\}$\\
 \hline 31 &$\{0_1,2_1,8_1,20_1\}$&$\{0_1,5_1,14_1,21_1\}$&$\{0_1,1_1,4_2\}$\\
 \hline\multirow{2}*{34} &$\{19_1,24_1,27_1,31_1\}$&$\{4_1,22_1,24_1,29_1\}$&$\{2_1,17_1,27_1,29_1\}$&$\{1_1,14_1,18_1,33_1\}$&$\{2_1,3_1,15_1,21_1\}$&$\{0_1,1_1,10_1,11_1\}$\\
  &$\{1_1,5_1,7_1,29_1\}$&$\{5_1,12_1,21_2\}$&$\{6_1,28_1,25_2\}$&$\{11_1,28_1,14_2\}$\\
\hline \multirow{3}*{40} &$\{4_1,12_1,15_1,27_1\}$&$\{22_1,25_1,33_1,37_1\}$&$\{3_1,11_1,12_1,17_1\}$&$\{4_1,16_1,17_1,30_1\}$&\bm{$\{0_1,10_1,20_1,30_1\}$}&\bm{$\{1_1,11_1,21_1,31_1\}$}\\
 &$\{4_1,20_1,38_1,39_1\}$&$\{2_1,6_1,18_1,39_1\}$&$\{0_1,6_1,29_1,38_1\}$&$\{1_1,2_1,15_1,19_1\}$&$\{8_1,15_1,17_1,33_1\}$&$\{13_1,19_1,32_1,34_1\}$\\
  &$\{20_1,37_1,2_2\}$&$\{6_1,35_1,11_2\}$&$\{26_1,35_1,33_2\}$&$\{9_1,16_1,12_2\}$\\
\end{longtable}
\end{footnotesize}

\subsection{The small optimal codes of length $n\equiv 2\pmod{3}$}\label{a3}

For each $n\in \{20$, $23$, $26$, $29$, $32$, $38$, $41$, $50$, $53$, $62$, $65$, $74$, $77$, $86$, $89$, $98$, $101$, $113\}$, the base codewords of the optimal codes in $\bbZ_n\times [2]$ are listed in Table~\ref{t3_B2} but with different group actions. The automorphism is $(0$  $\frac{n-8}{3}$  $\frac{2(n-8)}{3})(1$  $\frac{n-8}{3}+1 $ $\frac{2(n-8)}{3}+1)\cdots(\frac{n-8}{3}-1$  $\frac{2(n-8)}{3}-1$  $n-9)(n-8$ $n-6$  $n-4)(n-7$  $n-5$  $n-3)(n-2)(n-1)$ for $n=20$ and $23$; for $n=32$, the automorphism is
$(0$  $3$  $6$  $\cdots$  $21)(1$  $4$  $7$  $\cdots$  $22)(2$  $5$  $8$  $\cdots$  $23)(24$  $26$  $28$  $30$  $25$  $27$  $29$  $31)$;
for $n=38$, $50$, $62$, $74$, $86$, $98$, the automorphism is $(0$ $2$  $4$  $\cdots$  $n-4) (1$  $3$  $5$  $\cdots$  $n-3) (n-2$  $n-1)$; and for  $n=29$,  $41$, $53$, $65$, $77$, $89$, $101$, $113$, the automorphism is $(0$  $3$ $6$  $\cdots$  $n-5) (1$  $4$  $7$  $\cdots$  $n-4) (2$  $5$  $8$ $\cdots$  $n-3)(n-2)(n-1)$.

\setlength{\tabcolsep}{2pt}
\begin{footnotesize}
\begin{longtable}{l|llllll}
  \caption{Base codewords of small CWCs of length $n\equiv 2\pmod{3}$.}\label{t3_B2}\\
\hline
\multicolumn{1}{c|}{$n$} & \multicolumn{6}{c}{\text{Base Codewords}}\\\hline
\endfirsthead
\caption{Base codewords of small CWCs of length $n\equiv 2\pmod{3}$.}\\
\hline
\multicolumn{1}{c|}{$n$} & \multicolumn{6}{c}{\text{Base Codewords}}\\\hline
\endhead
\hline\\ \endfoot
\hline\endlastfoot
 \multirow{3}*{20} &\bm{$\{0_1,6_1,12_1,18_1\}$}&\bm{$\{1_1,7_1,13_1,19_1\}$}&\bm{$\{2_1,8_1,14_1,20_1\}$}&$\{13_1,18_1,22_1,30_1\}$&$\{6_1,11_1,29_1,30_1\}$&$\{2_1,5_1,10_1,29_1\}$\\
    &$\{12_1,13_1,25_1,28_1\}$&$\{5_1,16_1,20_1,28_1\}$&$\{2_1,3_1,16_1,19_1\}$&$\{0_1,2_1,9_1,28_1\}$&$\{0_1,3_1,10_1,11_1\}$&$\{5_1,24_1,15_2\}$\\
    &$\{27_1,29_1,13_2\}$&$\{0_1,22_1,20_2\}$&\bm{$\{24_2,25_2\}$}&\bm{$\{30_2,31_2\}$}\\
\hline \multirow{4}*{23} &$\{10_1,13_1,20_1,21_1\}$&$\{6_1,11_1,12_1,18_1\}$&\bm{$\{15_1,17_1,19_1,21_1\}$}&$\{6_1,9_1,10_1,19_1\}$&$\{3_1,14_1,18_1,19_1\}$&$\{1_1,9_1,15_1,18_1\}$\\
    &$\{3_1,9_1,20_1,22_1\}$&$\{0_1,7_1,18_1,20_1\}$&$\{1_1,3_1,8_1,10_1\}$&$\{5_1,7_1,10_1,15_1\}$&$\{2_1,7_1,13_1,14_1\}$&$\{2_1,4_1,6_1,21_1\}$\\
    &$\{15_1,22_1,2_2\}$&$\{3_1,15_1,6_2\}$&$\{10_1,14_1,4_2\}$&$\{1_1,22_1,0_2\}$&$\{7_1,19_1,8_2\}$&\bm{$\{15_2,16_2\}$}\\
    &\bm{$\{21_2,22_2\}$}\\
\hline
\multirow{4}*{26}&\bm{$\{3_1,9_1,15_1,21_1\}$}&\bm{$\{4_1,10_1,16_1,22_1\}$}&\bm{$\{5_1,11_1,17_1,23_1\}$}&\bm{$\{0_1,6_1,12_1,18_1\}$}&\bm{$\{1_1,7_1,13_1,19_1\}$}&\bm{$\{2_1,8_1,14_1,20_1\}$}\\
&$\{4_1,12_1,19_1,25_1\}$&$\{0_1,10_1,19_1,20_1\}$&$\{9_1,12_1,16_1,17_1\}$&$\{8_1,17_1,18_1,25_1\}$&$\{4_1,6_1,8_1,9_1\}$&$\{3_1,4_1,17_1,24_1\}$\\
&$\{7_1,9_1,20_1,24_1\}$&$\{2_1,5_1,9_1,18_1\}$&$\{2_1,7_1,10_1,21_1\}$&$\{1_1,21_1,12_2\}$&$\{0_1,17_1,1_2\}$&$\{1_1,5_1,22_2\}$\\
&$\{5_1,7_1,14_2\}$&$\{5_1,19_1,3_2\}$&$\{20_1,22_1,17_2\}$&\bm{$\{24_2,25_2\}$}\\
\hline
\multirow{2}*{29}&$\{6_1,14_1,16_1,28_1\}$&$\{2_1,8_1,17_1,24_1\}$&$\{3_1,16_1,19_1,25_1\}$&$\{2_1,3_1,6_1,7_1\}$&$\{0_1,2_1,19_1,27_1\}$&$\{1_1,13_1,14_1,17_1\}$\\
&$\{12_1,18_1,0_2\}$&$\{1_1,21_1,8_2\}$&$\{9_1,26_1,7_2\}$&\bm{$\{27_2,28_2\}$}\\
\hline
\multirow{3}*{32}&\bm{$\{0_1,6_1,12_1,18_1\}$}&\bm{$\{1_1,7_1,13_1,19_1\}$}&\bm{$\{2_1,8_1,14_1,20_1\}$}&$\{13_1,18_1,22_1,30_1\}$&$\{6_1,11_1,29_1,30_1\}$&$\{2_1,5_1,10_1,29_1\}$\\
    &$\{12_1,13_1,25_1,28_1\}$&$\{5_1,16_1,20_1,28_1\}$&$\{2_1,3_1,16_1,19_1\}$&$\{0_1,2_1,9_1,28_1\}$&$\{0_1,3_1,10_1,11_1\}$&$\{5_1,24_1,15_2\}$\\
    &$\{27_1,29_1,13_2\}$&$\{0_1,22_1,20_2\}$&\bm{$\{24_2,25_2\}$}&\bm{$\{30_2,31_2\}$}\\
\hline
\multirow{2}*{38}&$\{12_1,15_1,16_1,28_1\}$&$\{10_1,16_1,17_1,24_1\}$&$\{7_1,12_1,22_1,27_1\}$&$\{5_1,9_1,17_1,19_1\}$&$\{0_1,2_1,13_1,19_1\}$&\bm{$\{0_1,9_1,18_1,27_1\}$}\\
&$\{5_1,37_1,16_2\}$&$\{32_1,36_1,29_2\}$&\bm{$\{36_2,37_2\}$}\\
\hline \multirow{3}*{41}&$\{6_1,25_1,29_1,39_1\}$&$\{11_1,15_1,20_1,25_1\}$&$\{10_1,17_1,36_1,40_1\}$&$\{17_1,32_1,34_1,35_1\}$&$\{0_1,2_1,8_1,9_1\}$&$\{10_1,19_1,27_1,38_1\}$\\
&$\{2_1,14_1,24_1,27_1\}$&$\{0_1,4_1,6_1,27_1\}$&$\{4_1,10_1,15_1,22_1\}$&$\{7_1,21_1,6_2\}$&$\{10_1,13_1,23_2\}$&$\{20_1,28_1,4_2\}$\\
&\bm{$\{39_2,40_2\}$}\\
 \hline\multirow{2}*{50}&\bm{$\{0_1,12_1,24_1,36_1\}$}&\bm{$\{1_1,13_1,25_1,37_1\}$}&$\{14_1,20_1,25_1,36_1\}$&$\{2_1,17_1,31_1,47_1\}$&$\{28_1,35_1,37_1,45_1\}$&$\{17_1,24_1,37_1,43_1\}$\\
 &$\{23_1,27_1,36_1,44_1\}$&$\{0_1,2_1,20_1,23_1\}$&$\{5_1,6_1,10_1,20_1\}$&$\{23_1,48_1,22_2\}$&$\{26_1,49_1,3_2\}$&\bm{$\{48_2,49_2\}$}\\
\hline\multirow{3}*{53}&$\{16_1,23_1,24_1,50_1\}$&$\{28_1,42_1,43_1,46_1\}$&$\{12_1,15_1,31_1,50_1\}$&$\{5_1,11_1,25_1,41_1\}$&$\{0_1,12_1,22_1,46_1\}$&$\{12_1,30_1,39_1,44_1\}$\\
&$\{7_1,9_1,32_1,50_1\}$&$\{18_1,20_1,43_1,49_1\}$&$\{0_1,28_1,29_1,51_1\}$&$\{0_1,6_1,17_1,21_1\}$&$\{0_1,8_1,13_1,52_1\}$&$\{5_1,7_1,16_1,46_1\}$\\
&$\{43_1,47_1,3_2\}$&$\{12_1,32_1,19_2\}$&$\{2_1,44_1,5_2\}$&\bm{$\{51_2,52_2\}$}\\
\hline\multirow{3}*{62}&$\{3_1,15_1,28_1,37_1\}$&$\{0_1,16_1,48_1,50_1\}$&$\{13_1,31_1,33_1,52_1\}$&\bm{$\{0_1,15_1,30_1,45_1\}$}&$\{23_1,30_1,31_1,55_1\}$&$\{0_1,18_1,29_1,40_1\}$\\
  &$\{5_1,9_1,52_1,55_1\}$&$\{0_1,4_1,31_1,37_1\}$&$\{9_1,10_1,46_1,53_1\}$&$\{0_1,5_1,8_1,14_1\}$&$\{32_1,60_1,51_2\}$&$\{33_1,60_1,38_2\}$\\
  &\bm{$\{60_2,61_2\}$}\\
\hline \multirow{4}*{65}&$\{16_1,21_1,22_1,62_1\}$&$\{18_1,24_1,28_1,39_1\}$&$\{11_1,25_1,27_1,64_1\}$&$\{16_1,39_1,47_1,63_1\}$&$\{12_1,36_1,49_1,61_1\}$&$\{10_1,27_1,43_1,47_1\}$\\
&$\{9_1,26_1,31_1,41_1\}$&$\{12_1,17_1,21_1,55_1\}$&$\{7_1,20_1,59_1,62_1\}$&$\{1_1,9_1,16_1,23_1\}$&$\{2_1,11_1,12_1,31_1\}$&$\{28_1,37_1,55_1,56_1\}$\\
&$\{0_1,11_1,23_1,36_1\}$&$\{20_1,26_1,54_1,56_1\}$&$\{0_1,18_1,44_1,51_1\}$&$\{7_1,31_1,28_2\}$&$\{6_1,37_1,9_2\}$&$\{13_1,29_1,11_2\}$\\
&\bm{$\{63_2,64_2\}$}\\
\hline \multirow{3}*{74}&\bm{$\{1_1,19_1,37_1,55_1\}$}&$\{12_1,14_1,15_1,26_1\}$&$\{4_1,15_1,37_1,61_1\}$&$\{29_1,43_1,66_1,71_1\}$&$\{30_1,38_1,67_1,68_1\}$&\bm{$\{0_1,18_1,36_1,54_1\}$}\\
&$\{44_1,57_1,61_1,69_1\}$&$\{6_1,26_1,47_1,57_1\}$&$\{4_1,8_1,31_1,48_1\}$&$\{4_1,23_1,28_1,43_1\}$&$\{0_1,6_1,16_1,69_1\}$&$\{3_1,16_1,59_1,66_1\}$\\
&$\{10_1,17_1,55_1,57_1\}$&$\{36_1,73_1,62_2\}$&$\{9_1,72_1,15_2\}$&\bm{$\{72_2,73_2\}$}\\
\hline\multirow{4}*{77}&$\{33_1,51_1,53_1,72_1\}$&$\{23_1,35_1,62_1,72_1\}$&$\{18_1,24_1,28_1,43_1\}$&$\{25_1,30_1,65_1,76_1\}$&$\{23_1,30_1,46_1,57_1\}$&$\{20_1,24_1,48_1,70_1\}$\\
&$\{12_1,62_1,67_1,75_1\}$&$\{31_1,32_1,33_1,65_1\}$&$\{6_1,18_1,19_1,67_1\}$&$\{1_1,14_1,38_1,68_1\}$&$\{1_1,47_1,50_1,65_1\}$&$\{1_1,11_1,31_1,43_1\}$\\
&$\{10_1,19_1,42_1,57_1\}$&$\{3_1,10_1,14_1,45_1\}$&$\{6_1,29_1,35_1,51_1\}$&$\{5_1,7_1,58_1,64_1\}$&$\{1_1,4_1,32_1,45_1\}$&$\{8_1,66_1,69_1,74_1\}$\\
&$\{37_1,54_1,45_2\}$&$\{43_1,64_1,50_2\}$&$\{50_1,67_1,31_2\}$&\bm{$\{75_2,76_2\}$}\\
 \hline\multirow{3}*{86}&$\{42_1,53_1,73_1,75_1\}$&\bm{$\{0_1,21_1,42_1,63_1\}$}&$\{0_1,4_1,49_1,73_1\}$&$\{6_1,22_1,31_1,81_1\}$&$\{21_1,59_1,69_1,77_1\}$&$\{20_1,23_1,56_1,75_1\}$\\
 &$\{0_1,52_1,53_1,57_1\}$&$\{24_1,36_1,44_1,71_1\}$&$\{1_1,6_1,13_1,20_1\}$&$\{5_1,28_1,50_1,52_1\}$&$\{4_1,32_1,38_1,78_1\}$&$\{49_1,50_1,65_1,79_1\}$\\
 &$\{0_1,13_1,54_1,71_1\}$&$\{6_1,29_1,32_1,73_1\}$&$\{20_1,84_1,2_2\}$&$\{49_1,85_1,55_2\}$&\bm{$\{84_2,85_2\}$}\\
\hline\multirow{5}*{89}&$\{55_1,62_1,64_1,81_1\}$&$\{27_1,32_1,37_1,88_1\}$&$\{37_1,51_1,65_1,87_1\}$&$\{17_1,40_1,55_1,86_1\}$&$\{31_1,58_1,70_1,76_1\}$&$\{19_1,20_1,49_1,85_1\}$\\
&$\{37_1,42_1,61_1,84_1\}$&$\{30_1,54_1,60_1,76_1\}$&$\{30_1,63_1,64_1,67_1\}$&$\{16_1,49_1,51_1,86_1\}$&$\{20_1,34_1,53_1,59_1\}$&$\{4_1,12_1,38_1,48_1\}$\\
&$\{16_1,26_1,72_1,75_1\}$&$\{12_1,27_1,56_1,86_1\}$&$\{12_1,35_1,44_1,79_1\}$&$\{14_1,65_1,77_1,85_1\}$&$\{8_1,29_1,30_1,55_1\}$&$\{5_1,9_1,36_1,84_1\}$\\
&$\{14_1,39_1,46_1,59_1\}$&$\{5_1,21_1,32_1,39_1\}$&$\{0_1,2_1,13_1,17_1\}$&$\{6_1,56_1,53_2\}$&$\{18_1,27_1,76_2\}$&$\{3_1,79_1,24_2\}$\\
&\bm{$\{87_2,88_2\}$}\\
 \hline\multirow{4}*{98}&$\{13_1,71_1,72_1,75_1\}$&$\{28_1,80_1,85_1,95_1\}$&$\{20_1,42_1,46_1,73_1\}$&\bm{$\{0_1,24_1,48_1,72_1\}$}&\bm{$\{1_1,25_1,49_1,73_1\}$}&$\{22_1,41_1,71_1,90_1\}$\\
 &$\{6_1,24_1,82_1,89_1\}$&$\{5_1,21_1,26_1,62_1\}$&$\{3_1,17_1,53_1,59_1\}$&$\{26_1,40_1,72_1,82_1\}$&$\{28_1,45_1,62_1,63_1\}$&$\{35_1,46_1,67_1,87_1\}$\\
 &$\{4_1,6_1,12_1,75_1\}$&$\{0_1,9_1,80_1,93_1\}$&$\{1_1,9_1,24_1,54_1\}$&$\{8_1,67_1,69_1,95_1\}$&$\{1_1,52_1,64_1,75_1\}$&$\{37_1,96_1,8_2\}$\\
 &$\{68_1,97_1,61_2\}$&\bm{$\{96_2,97_2\}$}\\
\hline\multirow{5}*{101}&$\{24_1,38_1,43_1,50_1\}$&$\{26_1,50_1,72_1,84_1\}$&$\{32_1,49_1,93_1,100_1\}$&$\{51_1,64_1,75_1,84_1\}$&$\{31_1,44_1,53_1,93_1\}$&$\{22_1,59_1,77_1,78_1\}$\\
&$\{17_1,52_1,64_1,73_1\}$&$\{14_1,25_1,55_1,84_1\}$&$\{31_1,33_1,34_1,48_1\}$&$\{34_1,35_1,72_1,80_1\}$&$\{29_1,36_1,42_1,72_1\}$&$\{22_1,58_1,63_1,83_1\}$\\
&$\{0_1,21_1,72_1,89_1\}$&$\{15_1,17_1,43_1,47_1\}$&$\{10_1,42_1,60_1,64_1\}$&$\{54_1,59_1,65_1,79_1\}$&$\{20_1,43_1,59_1,92_1\}$&$\{0_1,16_1,73_1,83_1\}$\\
&$\{6_1,13_1,41_1,66_1\}$&$\{8_1,37_1,56_1,84_1\}$&$\{1_1,16_1,40_1,67_1\}$&$\{5_1,7_1,13_1,47_1\}$&$\{3_1,6_1,37_1,60_1\}$&$\{2_1,6_1,70_1,99_1\}$\\
&$\{44_1,47_1,0_2\}$&$\{29_1,44_1,8_2\}$&$\{76_1,84_1,94_2\}$&\bm{$\{99_2,100_2\}$}\\
\hline\multirow{6}*{113}&$\{58_1,91_1,95_1,101_1\}$&$\{10_1,29_1,45_1,100_1\}$&$\{12_1,34_1,81_1,87_1\}$&$\{57_1,72_1,101_1,109_1\}$&$\{60_1,64_1,71_1,106_1\}$&$\{22_1,49_1,107_1,109_1\}$\\
&$\{18_1,42_1,65_1,101_1\}$&$\{0_1,62_1,94_1,112_1\}$&$\{32_1,46_1,52_1,77_1\}$&$\{59_1,60_1,74_1,105_1\}$&$\{26_1,43_1,81_1,111_1\}$&$\{2_1,41_1,83_1,104_1\}$\\
&$\{13_1,27_1,47_1,94_1\}$&$\{27_1,60_1,65_1,77_1\}$&$\{28_1,30_1,48_1,60_1\}$&$\{11_1,22_1,38_1,45_1\}$&$\{27_1,43_1,48_1,61_1\}$&$\{38_1,67_1,82_1,95_1\}$\\
&$\{12_1,22_1,31_1,94_1\}$&$\{16_1,19_1,27_1,87_1\}$&$\{7_1,29_1,62_1,87_1\}$&$\{9_1,35_1,58_1,94_1\}$&$\{17_1,60_1,68_1,87_1\}$&$\{5_1,10_1,55_1,67_1\}$\\
&$\{17_1,27_1,30_1,55_1\}$&$\{3_1,12_1,66_1,73_1\}$&$\{5_1,8_1,45_1,84_1\}$&$\{17_1,110_1,21_2\}$&$\{93_1,94_1,95_2\}$&$\{47_1,71_1,1_2\}$\\
&\bm{$\{111_2,112_2\}$}\\
\end{longtable}
\end{footnotesize}

\subsection{The lower bound of $A_3(n,6,4)$ for $22$ special $n$ listed in Table~\ref{t7} }\label{a4}

\begin{lemma}\label{lowerbound}
  The lower and upper bound of $A_3(n,6,4)$ with $n\in M$ are given in Table~\ref{t7}, where $M=\{14$, $17$, $18$, $24$, $35$, $42$, $44$, $47$, $56$, $59$, $68$, $71$, $72$, $78$, $80$, $83$, $84$, $90$, $92$, $95$, $96$, $102\}$.
\end{lemma}

\begin{proof}
  Upper bounds are given by Lemma~\ref{A3(n,6,4)} and we only deal with the lower bounds. First, we consider the case of $n\equiv 0\pmod{3}$. For $n=18$ and $24$, we can show that $A_3(n,6,4)\geq 33$ and $55$, respectively, by computer search.  The corresponding codes are available upon request. For $n=42, 72, 78, 84, 90, 102$, the codes are constructed  from  $4$-GDDs of types $6^7$, $15^412^1$, $15^418^1$, $12^7$, $6^{15}$,  $6^{12}30^1$ combined with the required short codes, respectively. When $n=96$, the code is obtained by using a $4$-GDD of type $7^{12}10^1$ and a $(9,6,4)_3$ code with Property (B) by Theorem~\ref{24m}.

  The remaining cases are $n\equiv 1\pmod{3}$. $A_3(17,6,4)\geq 30$ and the corresponding code is found by computer.  For $n=14, 35, 44, 47, 56, 59, 68, 71, 80, 83, 92, 95$, the codes are constructed from $4$-GDDs of types $2^7$, $2^{12}11^1$, $2^{22}$, $2^{18}11^1$, $2^{24}8^1$, $2^{24}11^1$, $2^{24}20^1$, $2^{24}23^1$, $2^{27}26^1$, $2^{30}23^1$, $2^{33}26^1$, $2^{33}29^1$ combined with the required short codes, respectively.
\end{proof}


\vskip 10pt
\bibliographystyle{IEEEtran}
\bibliography{cwc}

\end{document}